\newif\ifextended
\extendedtrue

\documentclass[sigconf, nonacm]{acmart}
\usepackage{amsthm}
\usepackage{booktabs}
\usepackage{algorithm}
\usepackage{algpseudocode}
\usepackage{tikz}
\usetikzlibrary{arrows.meta, positioning, fit, backgrounds}
\usepackage{enumitem}
\usepackage{xspace}
\usepackage{graphicx}
\usepackage{subcaption}
\usepackage{pifont}
\usepackage{dblfloatfix}

\providecommand{\Fld}{\mathbb{F}}
\providecommand{\EF}{\mathbb{F}_{\mathrm{EF}}}
\providecommand{\EFld}{\mathbb{F}_{\mathrm{EF}}}
\providecommand{\Bool}{\{0,1\}}
\providecommand{\goldilocks}{2^{64}-2^{32}+1}
\providecommand{\supp}{\operatorname{supp}}

\providecommand{\Ntrain}{N}
\providecommand{\Ntest}{T}
\providecommand{\Ntables}{L}
\providecommand{\Kdepth}{K}
\providecommand{\knn}{k}
\providecommand{\Nclass}{C}
\providecommand{\Nbuckets}{M}

\providecommand{\draw}{d_{\mathrm{raw}}}
\providecommand{\dpca}{d_{\mathrm{pca}}}

\providecommand{\Vpca}{V_{\mathrm{pca}}}
\providecommand{\xraw}{x_{\mathrm{raw}}}

\providecommand{\harm}[1]{H_{#1}}
\providecommand{\phipos}{\phi^{+}}
\providecommand{\phineg}{\phi^{-}}

\providecommand{\phifinal}{\phi^{\mathrm{final}}}
\providecommand{\svsum}{\mathsf{sv\_sum}}
\providecommand{\svavg}{\mathsf{sv\_avg}}

\providecommand{\member}{\mathsf{member}}

\providecommand{\dotprod}{\mathsf{dp}}
\providecommand{\signbit}{\mathsf{sign}}

\providecommand{\hashkey}{\kappa}
\providecommand{\cnt}{\mathsf{cnt}}

\providecommand{\cnttr}{\mathsf{cnt\_c\_tr}}
\providecommand{\cntte}{\mathsf{cnt\_c\_te}}
\providecommand{\mhat}{\widehat{m}}
\providecommand{\mchat}{\widehat{m}_c}
\providecommand{\tchat}{\widehat{t}_c}

\providecommand{\MLE}[1]{\widetilde{#1}}
\providecommand{\MLExtrain}{\MLE{x_{\mathrm{tr}}}}
\providecommand{\MLExtest}{\MLE{x_{\mathrm{te}}}}
\providecommand{\eq}{\mathrm{eq}}
\providecommand{\cm}[1]{\mathsf{cm}(#1)}
\providecommand{\Comm}[1]{\mathsf{cm}(#1)}

\providecommand{\PCS}{\mathsf{PCS}}

\providecommand{\Setup}{\mathsf{Setup}}
\providecommand{\Verify}{\mathsf{Verify}}
\providecommand{\pp}{\mathsf{pp}}

\providecommand{\provider}{V}
\providecommand{\marketplace}{\mathcal{M}}

\providecommand{\commitment}[1]{\mathsf{C}_{#1}}
\providecommand{\indexset}[1]{S_{#1}}
\providecommand{\subphi}[1]{\Phi_{#1}}
\providecommand{\subD}[1]{D_{#1}}

\providecommand{\ptlow}{\tau_{\mathrm{low}}}

\providecommand{\parh}[1]{\noindent\textbf{#1}}

\newif\ifshowcomments
\showcommentstrue

\ifshowcomments
\newcommand{\sw}[1]{\mytodoblue{[wang: #1]}}
\else
\newcommand{\sw}[1]{}
\fi

\ifshowcomments
\newcommand{\pingchuan}[1]{\mytodored{[pingchuan: #1]}}
\else
\newcommand{\pingchuan}[1]{}
\fi

\ifshowcomments
\newcommand{\zhaoyu}[1]{\mytodoorange{[zhaoyu: #1]}}
\else
\newcommand{\zhaoyu}[1]{}
\fi

\ifshowcomments
\newcommand{\zy}[1]{\mytodoorange{[zhaoyu: #1]}}
\else
\newcommand{\zy}[1]{}
\fi

\ifshowcomments
\newcommand{\zhaoyuhigh}[1]{\mytodored{[zhaoyu\_high: #1]}}
\else
\newcommand{\zhaoyuhigh}[1]{}
\fi

\ifshowcomments
\newcommand{\zhaoyumedium}[1]{\mytodoorange{[zhaoyu\_medium: #1]}}
\else
\newcommand{\zhaoyumedium}[1]{}
\fi

\ifshowcomments
\newcommand{\zhaoyulow}[1]{\mytodocyan{[zhaoyu\_low: #1]}}
\else
\newcommand{\zhaoyulow}[1]{}
\fi

\newcommand{\mytodoblue}[1]{\textcolor{blue}{\ding{46}~{\sf}~#1}}
\newcommand{\mytodored}[1]{\textcolor{red}{\ding{46}~{\sf}~#1}}

\newcommand{\mytodoorange}[1]{\textcolor{orange}{\ding{46}~{\sf}~#1}}
\newcommand{\mytodocyan}[1]{\textcolor{cyan}{\ding{46}~{\sf}~#1}}

\newcommand{\F}{Fig.}

\newcommand{\T}{Table}

\newcommand{\Thm}{Theorem}

\newcommand{\appref}[1]{Appendix~\ref{#1}}
\newcommand{\Appref}[1]{\appref{#1}}

\newcommand{\rqThreeFullProofMB}{579}
\newcommand{\rqThreeFullProveS}{11.9}

\newcommand{\rqThreeFullVerifyS}{0.73}
\newcommand{\rqThreeNoSparsityProofMB}{579}
\newcommand{\rqThreeNoSparsityProveS}{13.4}

\newcommand{\rqThreeNoSparsityVerifyS}{0.86}
\newcommand{\rqThreeSparsityReductionPct}{12.6}
\newcommand{\rqThreeNoSuperOracleProofMB}{1247}
\newcommand{\rqThreeNoSuperOracleProveS}{14.7}
\newcommand{\rqThreeNoSuperOracleVerifyS}{1.06}
\newcommand{\rqThreeSuperOracleProofRatio}{2.2}
\newcommand{\rqThreeOriginalProofMB}{635}
\newcommand{\rqThreeOriginalProveS}{129.5}
\newcommand{\rqThreeOriginalVerifyS}{6.91}
\newcommand{\rqThreeBucketSpeedup}{11}

\newcommand{\rqThreeBreakdownLayersPct}{14.5}
\newcommand{\rqTwoDLowProveS}{20.0}
\newcommand{\rqTwoDHighProveS}{72.9}
\newcommand{\rqTwoNTopN}{65,536}
\newcommand{\rqTwoNTopProveS}{79.8}
\newcommand{\rqTwoCTopC}{100}
\newcommand{\rqTwoCTopProofMB}{1757}
\newcommand{\rqTwoCTopProveS}{32.2}
\newcommand{\rqTwoNtestSpreadPct}{6.6}

\newcommand{\rqOneBestSpeedup}{20.3}
\newcommand{\rqOneBestSpeedupDs}{cpu}

\newcommand{\rqOneKnnOomCount}{5}
\newcommand{\rqOneKnnOomTotal}{12}

\newcommand{\rqOneOrigOomCount}{3}
\newcommand{\rqOneOrigPeakSpeedup}{68.1}

\newcommand{\rqOneOrigMinSpeedup}{12.6}

\newcommand{\rqOneOrigOomList}{\texttt{vehicle}, \texttt{mnist}, \texttt{cifar10}}

\newcommand{\rqOneAbsDeltaAurocMax}{0.033}
\newcommand{\rqOneOpenMLProveSLo}{6.2}
\newcommand{\rqOneOpenMLProveSHi}{135.0}
\newcommand{\rqOneViTProveSLo}{1445}
\newcommand{\rqOneViTProveSHi}{2274}
\newcommand{\rqOneVerifySMax}{4.6}

\newcommand{\TableOneAUROCRows}{
        \texttt{phoneme} & 5,404 & 5 & 0.658 & 0.564 & 0.508 & \textbf{0.872} & \textbf{0.883} & 0.507 & 0.530 & 0.508 & \textbf{0.712} & \textbf{0.741} \\
        \texttt{2dplanes} & 10,000 & 10 & 0.668 & 0.515 & 0.491 & \textbf{0.943} & \textbf{0.938} & 0.517 & 0.514 & 0.492 & \textbf{0.780} & \textbf{0.726} \\
        \texttt{click} & 10,000 & 11 & 0.515 & 0.499 & 0.491 & \textbf{0.594} & \textbf{0.620} & 0.495 & 0.508 & 0.491 & \textbf{0.567} & \textbf{0.723} \\
        \texttt{wind} & 6,574 & 14 & 0.726 & 0.583 & 0.480 & \textbf{0.905} & \textbf{0.889} & 0.517 & 0.528 & 0.480 & \textbf{0.802} & \textbf{0.818} \\
        \texttt{cpu} & 8,192 & 21 & 0.756 & 0.569 & 0.511 & \textbf{0.959} & \textbf{0.952} & 0.499 & 0.541 & 0.513 & \textbf{0.844} & \textbf{0.860} \\
        \texttt{creditcard} & 10,000 & 23 & 0.540 & 0.531 & 0.491 & \textbf{0.713} & \textbf{0.724} & 0.496 & 0.521 & 0.492 & \textbf{0.686} & \textbf{0.773} \\
        \texttt{fraud} & 10,000 & 30 & 0.766 & 0.585 & 0.491 & \textbf{0.952} & \textbf{0.941} & 0.508 & 0.551 & 0.494 & \textbf{0.826} & \textbf{0.795} \\
        \texttt{pol} & 10,000 & 48 & 0.636 & 0.547 & 0.491 & \textbf{0.970} & \textbf{0.957} & 0.499 & 0.520 & 0.491 & \textbf{0.778} & \textbf{0.880} \\
        \texttt{vehicle} & 10,000 & 100 & 0.569 & 0.551 & 0.491 & \textbf{0.853} & \textbf{0.820} & 0.498 & 0.517 & 0.492 & \textbf{0.779} & \textbf{0.771} \\
        \texttt{apsfail} & 10,000 & 170 & 0.783 & 0.539 & 0.491 & \textbf{0.971} & \textbf{0.964} & 0.517 & 0.541 & 0.494 & \textbf{0.916} & \textbf{0.840} \\
        \texttt{mnist} & 60,000 & 768 & 0.747 & -- & -- & \textbf{0.982} & \textbf{0.979} & 0.519 & -- & -- & \textbf{0.855} & \textbf{0.843} \\
        \texttt{cifar10} & 50,000 & 768 & 0.810 & -- & -- & \textbf{0.995} & \textbf{0.998} & 0.504 & -- & -- & \textbf{0.950} & \textbf{0.899} \\
}

\newcommand{\TableOneTimeRows}{
        \texttt{phoneme} & 0.86 & 10.48 & 3.54 & 8.74 & \textbf{0.16} \\
        \texttt{2dplanes} & 0.88 & 47.27 & 8.46 & 28.59 & \textbf{0.86} \\
        \texttt{click} & 0.74 & 58.91 & 9.34 & 24.70 & \textbf{0.31} \\
        \texttt{wind} & 0.43 & 9.20 & 6.54 & 10.91 & \textbf{0.31} \\
        \texttt{cpu} & 0.57 & 30.43 & 9.74 & 17.50 & \textbf{0.19} \\
        \texttt{creditcard} & 1.16 & 73.07 & 12.51 & 25.89 & \textbf{0.54} \\
        \texttt{fraud} & 0.69 & 79.65 & 21.76 & 24.97 & \textbf{0.20} \\
        \texttt{pol} & 0.77 & 27.40 & 17.55 & 27.10 & \textbf{0.23} \\
        \texttt{vehicle} & 0.86 & 398.04 & 14.25 & 28.99 & \textbf{0.78} \\
        \texttt{apsfail} & 0.76 & 18.20 & 20.42 & 28.29 & \textbf{0.47} \\
        \texttt{mnist} & 9.08 & $>$1h & $>$1h & 729.99 & \textbf{3.25} \\
        \texttt{cifar10} & 7.22 & $>$1h & $>$1h & 721.97 & \textbf{4.21} \\
}

\newcommand{\TableOneZKPRows}{
        \texttt{phoneme} & 5,404 & 5 & 32.1 & 2.23 & 963 & 80.2 & 1.42 & 498 & 6.16 & 0.75 & 356 \\
        \texttt{2dplanes} & 10,000 & 10 & 158 & 3.70 & 1389 & 639 & 8.85 & 942 & 21.7 & 0.79 & 469 \\
        \texttt{click} & 10,000 & 11 & 177 & 4.45 & 1389 & 302 & 2.54 & 742 & 11.3 & 0.68 & 423 \\
        \texttt{wind} & 6,574 & 14 & 43.3 & 2.06 & 988 & 158 & 2.72 & 633 & 11.8 & 0.66 & 413 \\
        \texttt{cpu} & 8,192 & 21 & 266 & 7.03 & 1512 & 296 & 2.54 & 751 & 13.1 & 0.84 & 440 \\
        \texttt{creditcard} & 10,000 & 23 & 280 & 5.19 & 1512 & 627 & 9.63 & 950 & 23.2 & 0.92 & 487 \\
        \texttt{fraud} & 10,000 & 30 & 277 & 5.67 & 1515 & 332 & 4.35 & 751 & 13.9 & 0.76 & 452 \\
        \texttt{pol} & 10,000 & 48 & OOM & OOM & OOM & 325 & 3.62 & 843 & 25.7 & 0.89 & 544 \\
        \texttt{vehicle} & 10,000 & 100 & OOM & OOM & OOM & OOM & OOM & OOM & 135 & 1.43 & 680 \\
        \texttt{apsfail} & 10,000 & 170 & OOM & OOM & OOM & 735 & 9.59 & 958 & 10.8 & 0.75 & 499 \\
        \texttt{mnist} & 60,000 & 768 & OOM & OOM & OOM & OOM & OOM & OOM & 1445 & 3.71 & 1164 \\
        \texttt{cifar10} & 50,000 & 768 & OOM & OOM & OOM & OOM & OOM & OOM & 2274 & 4.55 & 1332 \\
}

\theoremstyle{definition}
\newtheorem{example}{Example}[section]

\newtheorem{theorem}{Theorem}[section]

\newcommand{\tool}{\textsc{ZK-Value}\xspace}
\newcommand{\zkls}{\textsc{ZK-LSH-Shapley}\xspace}
\newcommand{\zklsnaive}{\textsf{naive-ZK-LSH-Shapley}\xspace}

\begin{document}

\title{\tool: A Practical Zero-Knowledge System for Verifiable Data Valuation}

\author{Zhaoyu Wang}
\affiliation{
  \institution{HKUST}
  \city{Kowloon}
  \country{Hong Kong SAR}
}

\author{Pingchuan Ma}
\authornote{Corresponding author (email: pma@zjut.edu.cn).}
\affiliation{
  \institution{Zhejiang University of Technology}
  \city{Hangzhou}
  \country{China}
}

\author{Zhantong Xue}
\affiliation{
  \institution{HKUST}
  \city{Kowloon}
  \country{Hong Kong SAR}
}

\author{Yuguang Zhou}
\affiliation{
  \institution{HKUST}
  \city{Kowloon}
  \country{Hong Kong SAR}
}

\author{Qixin Zhang}
\affiliation{
  \institution{Nanyang Technological University}
  \city{}
  \country{Singapore}
}

\author{Xiaoqin Zhang}
\affiliation{
  \institution{Zhejiang University of Technology}
  \city{Hangzhou}
  \country{China}
}

\author{Shuai Wang}
\affiliation{
  \institution{HKUST}
  \city{Kowloon}
  \country{Hong Kong SAR}
}

\begin{abstract}
Data valuation is a foundational task in data marketplaces, where a Shapley-value attribution determines how a buyer's payment is distributed among data providers. Typically, the marketplace operator runs this attribution alone, requiring participants and external auditors to trust scores they cannot independently recompute on the underlying private data. 
While zero-knowledge proofs (ZKPs) can theoretically reconcile this conflict between privacy and verifiability, existing ZK valuation systems fail to scale to real-world marketplace demands due to prohibitive proving times or the requirement to disclose validation cohorts. 

We present ZK-VALUE, a practical, end-to-end ZK data-valuation system. Our solution bridges the scalability gap through a fully co-designed architecture: (1) LSH-Shapley, a locality-based valuation primitive that replaces expensive pairwise distance metrics with per-bucket collision counts; (2) ZK-LSH-SHAPLEY, a tailored ZKP protocol that drastically reduces witness size by encoding these counts into bucket-level histograms rather than naive per-pair tensors; and (3) structural proof-system optimizations, specifically super-oracle batching and sparsity skipping. Evaluated across 12 standard datasets, \tool delivers valuation quality on par with state-of-the-art baselines (within $\rqOneAbsDeltaAurocMax$ AUROC of exact KNN-Shapley), while generating proofs in seconds to minutes and outperforming specialized ZK baselines by $\rqOneOrigMinSpeedup\times$ to $\rqOneOrigPeakSpeedup\times$ in proving time, with verification in under $\rqOneVerifySMax$\,s.

\end{abstract}

\maketitle

\section{Introduction}
\label{sec:intro}

Training data is now a paid commodity traded through \emph{data
marketplaces}~\cite{koutris2015query,deep2017qirana,chen2019towards,liu2021dealer}.
A typical marketplace~\cite{fernandez2020data,pei2023data}  brings together \emph{data providers}, a \emph{buyer} who
pays for aggregated data to train a downstream model, a \emph{marketplace
operator} who runs the transaction, and \emph{external auditors} who oversee
compliance. Every transaction hinges on \emph{data
valuation}~\cite{sim2022data,yoon2020data}: how should payment be divided among providers so
that each share reflects the value its data contributes to the buyer's
downstream task? To date, \emph{Shapley
value}~\cite{ghorbani2019data,10.14778/3342263.3342637,wang2023privacy,kwon2022beta,wang2023data}
has become the de facto standard for answering this question in a principled
way. Each provider is credited in proportion to the marginal gain its data
brings, averaged over all coalitions and scored on the buyer's \emph{validation
cohort}. As a well-known limitation of the Shapley value, however, computing it
directly is often intractable; the community has therefore developed a variety
of ML-specific surrogates, e.g.,
\emph{KNN-Shapley}~\cite{10.14778/3342263.3342637}, that are tractable at modern
training-set sizes and feature dimensions.

Despite the progress in efficiency, however, the Shapley value and its
surrogates remain technically unverifiable in the data valuation setting. The
marketplace operator, who runs the valuation and earns a fee on each
transaction, has a direct financial motive to skew the payout, and nothing in
the setup prevents it from substituting a cheaper approximation, steering
payouts toward a preferred counter-party, or delivering data different from what
the valuation actually saw. Each non-operator party is thus left with an
unauditable concern, and no party can settle that concern without the raw-data
access that privacy forbids. The claimed scores are accepted in practice yet
remain technically unverifiable, a fundamental conflict between \emph{privacy}
and \emph{verifiability} at the heart of modern data marketplaces.

\begin{example}[Hospital Consortium Data Marketplace]
\label{ex:hospital}
Consider a consortium of hospitals selling de-identified patient data to a
pharmaceutical buyer training a drug-response classifier. The operator splits
the buyer's \$1M payment across hospitals according to Shapley values scored on
the buyer's validation cohort. Suppose the operator quietly shifts 5\% of the
payout from a small rural hospital to a large client hospital by tampering with
score entries. As Figure~\ref{fig:marketplace_gap} illustrates, no other party
can catch the shift: the rural hospital sees only its reduced payment with no
way to verify the correct amount, the regulator sees the published scores but is
barred from patient records needed to recompute, and the buyer has no standing
to audit payout allocations. The only recourse, i.e., recomputing scores on raw
patient data, is exactly what privacy policy and regulations forbid. Breaking
this deadlock demands \emph{privacy-preserving and verifiable data valuation}.
\end{example}

\begin{figure}[t]
  \centering
  \resizebox{1\linewidth}{!}{\providecommand{\concernline}[1]{{\scriptsize\itshape\color{red!62!black}#1}}
\providecommand{\stepone}{\textcolor{black!72}{\ding{172}}}
\providecommand{\steptwo}{\textcolor{black!72}{\ding{173}}}

\resizebox{\columnwidth}{!}{
\begin{tikzpicture}[
    font=\sffamily\scriptsize,
    every node/.append style={font=\sffamily\scriptsize},
    party/.style       = {draw, rounded corners=2.8pt, align=center,
                          minimum height=1.30cm, minimum width=1.80cm,
                          inner xsep=2pt, inner ysep=1.5pt,
                          line width=0.85pt, fill=white,
                          font=\sffamily\footnotesize},
    provider/.style    = {party, draw=teal!58!black,
                          top color=teal!22, bottom color=teal!3},
    buyer/.style       = {party, draw=violet!55!black,
                          top color=violet!22, bottom color=violet!3},
    auditor/.style     = {party, draw=orange!65!black,
                          top color=orange!22, bottom color=orange!3},
    operator/.style    = {draw, rounded corners=5pt, align=center,
                          minimum height=1.40cm, minimum width=2.30cm,
                          inner sep=2pt, line width=1.3pt,
                          draw=blue!58!black,
                          top color=blue!38, bottom color=blue!8},
    inflow/.style      = {-{Latex[length=1.9mm, width=1.5mm]},
                          line width=0.85pt, color=black!72},
    outflow/.style     = {-{Latex[length=1.9mm, width=1.5mm]},
                          line width=0.95pt, color=black!72,
                          dash pattern=on 3pt off 1.5pt},
    flabel/.style      = {font=\sffamily\scriptsize, pos=0.4, sloped,
                          fill=white, inner xsep=2pt, inner ysep=0.5pt,
                          text=black!78},
    rlabel/.style      = {font=\sffamily\scriptsize, pos=0.6, sloped,
                          fill=white, inner xsep=2pt, inner ysep=0.5pt,
                          text=black!78},
  ]

  \node[operator] (op) at (0,0) {
    \textbf{\footnotesize Marketplace}\\ \textbf{\footnotesize Operator}\\[-0.5pt]
    \textcolor{red!62!black}{\scriptsize\textit{financial motive to}} \\ \textcolor{red!62!black}{\scriptsize\textit{skew the payout}}
  };

  \node[provider] (pa) at (-2.4, 1.3)
    {\textbf{Provider A}\\[-0.5pt]
     {\scriptsize(rural hospital)}\\[-1pt]
     \concernline{fair share?}};
  \node[provider] (pb) at (-2.4,-1.3)
    {\textbf{Provider B}\\[-0.5pt]
     {\scriptsize(client hospital)}\\[-1pt]
     \concernline{fair share?}};
  \node[buyer] (by) at ( 2.4, 1.3)
    {\textbf{Buyer}\\[-0.5pt]
     {\scriptsize(pharma)}\\[-1pt]
     \concernline{honest score??}};
  \node[auditor] (au) at ( 2.4,-1.3)
    {\textbf{Ext.\ Auditor}\\[-0.5pt]
     {\scriptsize(regulator / press)}\\[-1pt]
     \concernline{no raw access}};

  \draw[inflow]  (pa.east)        to[bend left=22]
        node[flabel]{$D^A_{\mathrm{tr}}$}    (op.140);
  \draw[outflow] (op.160)          to[bend left=22]
        node[rlabel,pos=0.4]{payout A}             (pa.south);

  \draw[inflow]  (pb.east)        to[bend right=22]
        node[flabel]{$D^B_{\mathrm{tr}}$}    (op.220);
  \draw[outflow] (op.200)          to[bend right=22]
        node[rlabel,pos=0.4]{payout B}             (pb.north);

  \draw[inflow]  (by.west)         to[bend right=22]
        node[flabel, align=center]{$D_{\mathrm{val}}$\\ \$1M} (op.45);
  \draw[outflow] (op.15)           to[bend right=22]
        node[rlabel,pos=0.5]{$D_{\mathrm{tr}}$}             (by.south);

  \draw[outflow] (op.315) to[bend right=22]
        node[rlabel,pos=0.4]{score}           (au.west);
\end{tikzpicture}
}}
  \caption{The verifiability gap in a data marketplace. }
  \label{fig:marketplace_gap}
\end{figure}
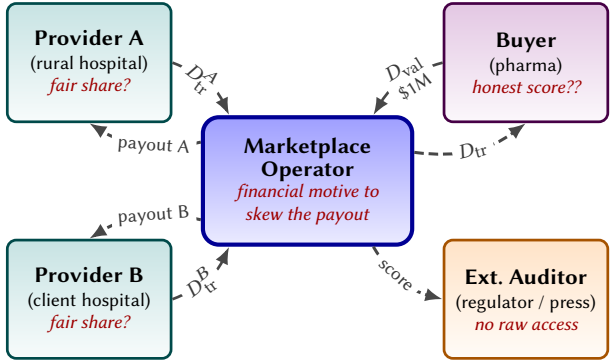

\parh{Zero-knowledge proofs as the bridge.}~As a cryptographic primitive,
zero-knowledge proofs (ZKPs)~\cite{thaler2022proofs, goldwasser1989knowledge,
bitansky2017hunting} offer a principled way to reconcile privacy with
verifiability. By allowing a prover to demonstrate that a computation was
executed correctly on hidden inputs without revealing the inputs themselves,
ZKPs have already seen successful deployment in important areas like ML
inference \cite{chen2024zkml}, SQL query processing
\cite{li2023zksql,gu2025poneglyphdb}, and data analytics \cite{wang2025privacy}.

However, a cryptographic building block alone does not make a practical marketplace protocol. Adapting this primitive to a multi-party data marketplace dictates four system-level properties that any deployable verifiable data-valuation protocol must jointly satisfy. \emph{(P-1) Public verifiability}: the protocol produces
a single non-interactive proof that any third party verifies
locally. \emph{(P-2) Data privacy with commitment binding}: third-party
verifiers learn nothing about the private inputs beyond what the
released scores reveal, and commitments bind each party to the
data it contributed. \emph{(P-3)
Operator-as-prover architecture}: the prover role can only be filled by a party that
both holds every valuation input and operates infrastructure
capable of amortizing ZK proving across transactions, where only the operator satisfies both
conditions. \emph{(P-4) Practical cost at modern ML scale}: proving must
remain tractable at modern ML training scale, and verification
cheap enough for any third party to check the proof quickly.

\parh{Limitations of existing solutions.} \T~ \ref{tab:overview} shows three representative families each fail at least one requirement (P-1)--(P-4). (1) DP/MPC methods \cite{wang2023privacy,tian2022private,peng2025reliable} hide raw data but lack non-interactive correctness proofs for third-party verification, failing (P-1). (2) ZK-DV \cite{liu2026bridging}, a concurrent work and the only other ZK data-valuation system, forces buyers to disclose validation data to sellers (partial P-2 failure) and requires sellers to run individual provers, forfeiting the marketplace cost-sharing of a shared operator (failing P-3) (3) ZK encodings of existing Shapley algorithms \cite{risc0,sp1} inherit prohibitive computational costs (failing P-4); for example, KNN-Shapley \cite{10.14778/3342263.3342637,wang2023privacy} requires $\Ntrain\!\cdot\!\Ntest\!\cdot\! d$ work plus a top-$k$ sort, while retraining \cite{ghorbani2019data,kwon2022beta,wang2023data} and gradient variants \cite{koh2017understanding} scale even worse. Satisfying all four properties simultaneously requires co-designing the valuation algorithm with its proof system, which is the approach \tool takes.

\begin{table}[t]
  \centering
    \caption{Requirements (P-1)--(P-4) coverage by solution families and \tool. \textcolor{green!55!black}{\ding{51}} = met, \textcolor{red!75!black}{\ding{55}} = not met, \textcolor{orange!85!black}{$\circ$} = partial.}
  \label{tab:overview}
  \resizebox{0.85\linewidth}{!}{
  \begin{tabular}{@{}lcccc@{}}
    \toprule
    \textbf{Solution family} & \textbf{P-1} & \textbf{P-2} & \textbf{P-3} & \textbf{P-4} \\
    \midrule
    DP/MPC valuation                & \textcolor{red!75!black}{\ding{55}} & \textcolor{green!55!black}{\ding{51}} & \textcolor{red!75!black}{\ding{55}} & \textcolor{red!75!black}{\ding{55}} \\
    ZK-DV       & \textcolor{green!55!black}{\ding{51}} & \textcolor{orange!85!black}{$\circ$}  & \textcolor{red!75!black}{\ding{55}} & \textcolor{red!75!black}{\ding{55}} \\
    ZK encoding of existing Shapley & \textcolor{green!55!black}{\ding{51}} & \textcolor{orange!85!black}{$\circ$}  & \textcolor{green!55!black}{\ding{51}} & \textcolor{red!75!black}{\ding{55}} \\
    \tool (ours)                    & \textcolor{green!55!black}{\ding{51}} & \textcolor{green!55!black}{\ding{51}} & \textcolor{green!55!black}{\ding{51}} & \textcolor{green!55!black}{\ding{51}} \\
    \bottomrule
  \end{tabular}
  }

\end{table}

\parh{Our Solution.}~We present \tool, the first practical ZK
data-valuation system deployable in a real marketplace. \tool adopts a multi-party architecture: providers and the buyer publish hiding-and-binding commitments to their datasets, while a marketplace operator computes valuation scores and produces a single non-interactive zero-knowledge proof attesting that the scores are consistent with all commitments.

To make proving costs practical at modern ML scale, we co-design the valuation algorithm and proof system.We introduce LSH-Shapley (\S\ref{sec:method}), which hashes data into locality-sensitive buckets and computes Shapley values from bucket histograms, avoiding explicit pairwise distance computations.
We pair this with \zkls{} (\S\ref{sec:method:bucket}), a tailored ZKP protocol commits only to per-bucket histograms and proves scores directly against them. With optimizations such as super-oracle batching and sparsity skipping (\S\ref{sec:optimizations}), this system enforces input integrity, satisfies (P-1)--(P-4), and reduces per-valuation proving time to a practical seconds-to-minutes range.

We evaluate \tool on 12 standard valuation datasets at realistic marketplace
sizes. On valuation quality, \tool tracks exact KNN-Shapley within
$\rqOneAbsDeltaAurocMax$ AUROC across all datasets while outperforming
gradient-based and Monte-Carlo Shapley baselines. On end-to-end proving cost,
\tool achieves $\rqOneOrigMinSpeedup\times$ to $\rqOneOrigPeakSpeedup\times$
speedups over specialized ZK baselines for KNN- and LSH-Shapley, proves a full
valuation in seconds on tabular workloads with verification under
$\rqOneVerifySMax$\,s, and is the only system that scales to high-dimensional
ViT-embedding workloads.

\parh{Contributions.}~In summary, our contributions are:
\begin{itemize}[leftmargin=*]
  \item \textbf{Conceptually,} we identify and formulate
  \emph{privacy-preserving and verifiable data valuation} as a first-class
  problem arising in cross-organizational data marketplaces, and distill a four
  requirements ((P-1)--(P-4)) for practical solutions.
  \item \textbf{Technically,} we present \tool, a multi-party ZK data valuation system with an operator-as-prover architecture, built on a co-design of the valuation algorithm and underlining proof system: (i) \emph{LSH-Shapley}, a Shapley-family algorithm that replaces KNN-Shapley's pairwise-distance   
cost with per-bucket collision counts, and (ii) \zkls{}, a specialized ZK protocol that certifies LSH-Shapley scores efficiently, together with 
two proof-system optimizations (super-oracle batching and sparsity skipping) that further reduce proof size and prover time.
  \item \textbf{Empirically,} our evaluation across 12 datasets demonstrates that \tool preserves strong valuation quality (within $\rqOneAbsDeltaAurocMax$ AUROC of exact KNN-Shapley) while reducing ZK proving time by $\rqOneOrigMinSpeedup\times$ to $\rqOneOrigPeakSpeedup\times$ over specialized ZK baselines and verifying in under $\rqOneVerifySMax$\,s.

\end{itemize}

\section{Background}
\label{sec:background}

This section recaps the three ingredients \tool rests on:
Shapley-value data valuation (\S\ref{sec:bg:shapley}),
zero-knowledge SNARKs (\S\ref{sec:bg:zkp}), and
locality-sensitive hashing (\S\ref{sec:bg:lsh}).

\subsection{Data Valuation via Shapley Values}
\label{sec:bg:shapley}

\parh{Shapley value.}~Given a utility function
$u : 2^{[\Ntrain]} \to \mathbb{R}$ that scores any subset of the
$\Ntrain$ training points (e.g., the accuracy the buyer's model
attains on the validation cohort after training on that subset), the
\emph{Shapley value} $\phi_i$ of point $i$ is its weighted average
marginal contribution across every coalition of the other points~\cite{shapley1953value}:
\[
  \phi_i
  \;=\;
  \sum_{S \subseteq [\Ntrain] \setminus \{i\}}
  \frac{|S|!\,(\Ntrain - |S| - 1)!}{\Ntrain!}
  \bigl(u(S \cup \{i\}) - u(S)\bigr).
\]
$\phi_i$ is uniquely characterized by four fairness axioms
(efficiency, symmetry, dummy-player,
linearity)~\cite{sundararajan2020many}. Direct evaluation enumerates
$2^{\Ntrain - 1}$ coalitions and retrains the utility on each, so
it is $\#$P-hard for a general utility and infeasible at realistic
marketplace sizes. This has motivated \emph{utility-specific closed
forms} that sidestep the enumeration entirely.

\parh{KNN-Shapley.}~
KNN-Shapley is the standard tractable proxy in this family~\cite{10.14778/3342263.3342637,pandl2021trustworthy}.
It instantiates the utility as a $\knn$-nearest-neighbor classifier scored on a held-out validation set. This approach elegantly collapses the combinatorial Shapley sum into a structure that depends solely on how training points rank in distance to a validation query. While this provides a closed-form solution computable in $O(\Ntrain \log \Ntrain)$ time per validation point, the underlying computation is inherently ill-suited for ZKP encoding. We detail and resolve this challenge in \S\ref{sec:system:challenge}.

\subsection{Zero-Knowledge Proofs}
\label{sec:bg:zkp}

\parh{Setting and security properties.}~A ZKP allows a prover
$\mathsf{P}$ to convince a verifier $\mathsf{V}$ that a statement
about a private witness is true without revealing anything beyond
the statement's truth~\cite{goldwasser1989knowledge}. A ZKP protocol must satisfy three standard
properties~\cite{bitansky2017hunting}: \emph{completeness} (an honest prover always produces
a verifying proof), \emph{soundness} (no malicious prover can
convince the verifier of a false statement except with negligible
probability), and \emph{zero-knowledge} (the proof leaks nothing
about the private witness beyond what is implied by the
statement).

\parh{Primitives we compose.}~\tool builds on three standard ZK
primitives; for each we state what it does, then the formal
guarantee we use, deferring cryptographic details to the cited
references.

\emph{Sumcheck and multilinear extensions.}~The sumcheck protocol~\cite{thaler2022proofs,lund1992algebraic} enables a verifier to validate the sum of a polynomial $g$ over $\Bool^n$ by checking a single evaluation of its multilinear extension at a random point. The protocol requires $n$ rounds with a soundness error of at most $n \cdot \deg(g) / |\Fld|$ via the Schwartz--Zippel lemma. Since prover complexity scales with the degree of $g$, sumcheck cost is governed by the summand's algebraic structure. \tool\ employs sumcheck to enforce bucket-count consistency checks (\S\ref{sec:method}).

\emph{Polynomial commitments.}~A polynomial commitment scheme~\cite{wahby2018doubly} (PCS) allows a prover to publish a succinct digest $\Comm{f}$ and subsequently prove the correctness of evaluations $f(z)$ without revealing the full polynomial. We instantiate our PCS using a Brakedown-style construction~\cite{golovnev2023brakedown} over the Goldilocks field ($p = \goldilocks$). This construction incurs an opening cost of $O(\sqrt{M_{\mathrm{oracle}}})$, a square-root scaling factor that motivates the optimizations discussed in \S\ref{sec:optimizations}.

\emph{Lookup arguments.}~Lookup arguments~\cite{gabizon2020plookup} efficiently verify that a set of committed values exists within a witness table (e.g., verifying bucket counts within $[0, M_{\max}]$). Following the LogUp approach~\cite{thaler2022proofs}, we utilize log-derivative identities to reduce multiple lookups to a single sumcheck instance. This amortizes $m$ checks against a table $T$ into $O(\log |T| + \log m)$ constraints, providing a significant efficiency gain over naive encodings.

\parh{From interactive to non-interactive.}~The three primitives
above are natively interactive. The Fiat--Shamir transform~\cite{fiat1986prove}
collapses any such interactive protocol into a single
\emph{succinct non-interactive argument of knowledge} (SNARK) by
deriving each verifier challenge from a cryptographic hash of the
transcript so far, yielding one static proof any third party can
verify offline.

\subsection{Locality-Sensitive Hashing}
\label{sec:bg:lsh}

\parh{Sensitive hash families.}~Locality-sensitive hashing
(LSH)~\cite{gionis1999similarity,lv2007multi,datar2004locality} is a family of randomized hash
functions under which similar inputs collide with
higher probability than dissimilar ones. A hash family
$\mathcal{H}$ is $(r_1, r_2, p_1, p_2)$-\emph{sensitive} for a
distance $d(\cdot, \cdot)$ if
\[
  \Pr_{h \in \mathcal{H}}\!\bigl[h(x) = h(y)\bigr] \geq p_1
  \text{ when } d(x, y) \leq r_1,
\]
\[
  \Pr_{h \in \mathcal{H}}\!\bigl[h(x) = h(y)\bigr] \leq p_2
  \text{ when } d(x, y) \geq r_2,
\]
with $p_1 > p_2$. For angular similarity on $\mathbb{R}^d$, the
\emph{SimHash}
family~\cite{charikar2002similarity} is
$\{\,h_r : x \mapsto \signbit(\langle r, x \rangle)
    \mid r \sim \mathcal{N}(0, I_d)\,\}$,
with single-bit collision probability
$1 - \theta(x, y)/\pi$ where $\theta(x, y)$ is the angle between
$x$ and $y$. \tool's LSH-Shapley protocol (\S\ref{sec:method:plaintext}) instantiates its hash family with SimHash.

\parh{Bucket structure.}~A single-bit hash is too coarse, so LSH
stacks $\Kdepth$ independent SimHash bits per table into a bucket
id (AND-amplification drops the collision probability to
$(1 - \theta/\pi)^{\Kdepth}$) and runs $\Ntables$ independent
tables in parallel (OR-amplification). Together
$(\Ntables, \Kdepth)$ control bucket granularity and recall:
larger $\Kdepth$ grows the bucket count $\Nbuckets = 2^{\Kdepth}$
and sharpens discrimination, while larger $\Ntables$ averages out
per-table variance.

\section{System Overview}
\label{sec:system}

This section specifies the parties, commitments, and protocol phases that realize \tool (\F~\ref{fig:protocol_architecture}), and states the threat model, security guarantees, and application scope under which the construction is analyzed.

\subsection{System Model}
\label{sec:system:model}

\parh{Parties.}~A \tool deployment involves four parties. \emph{Data
providers} $\provider_1, \dots, \provider_m$ each hold a private
training shard $\subD{i}$ whose rows occupy positions $\indexset{i} \subseteq [\Ntrain]$ in the aggregated training set $\subD{} = \bigcup_i \subD{i}$ of size $\Ntrain$; the index sets $\{\indexset{i}\}_{i \in [m]}$ partition $[\Ntrain]$ and later identify each provider's rows during selective opening. The \emph{buyer} $B$ holds a private validation set
$\subD{\mathrm{val}}$ of size $\Ntest$, which encodes the
downstream task. The \emph{marketplace operator} $\marketplace$
runs the valuation on behalf of providers and the buyer. \emph{External verifiers} (regulators, non-participating providers, investigative journalists) participate only through the protocol's public information and hold no private witnesses or per-provider scores.

\parh{Commitments.}~Each provider $\provider_i$ submits
$\subD{i}$ to $\marketplace$ for valuation and publishes
$\commitment{i} = \Comm{\subD{i}}$, where $\Comm{\cdot}$ is the cryptographic commitment scheme of \S\ref{sec:bg:zkp}: binding so that the prover cannot later open against a different witness, and hiding so that $\commitment{i}$ leaks nothing about $\subD{i}$. The buyer does the same with $\subD{\mathrm{val}}$ and $\commitment{\mathrm{val}} = \Comm{\subD{\mathrm{val}}}$. The
commitments are the public references the rest of the protocol
operates against: every claim about the released scores is later
checked against them, so the prover cannot quietly compute on a
different witness.

\begin{figure}[t]
  \centering
  \includegraphics[width=\columnwidth]{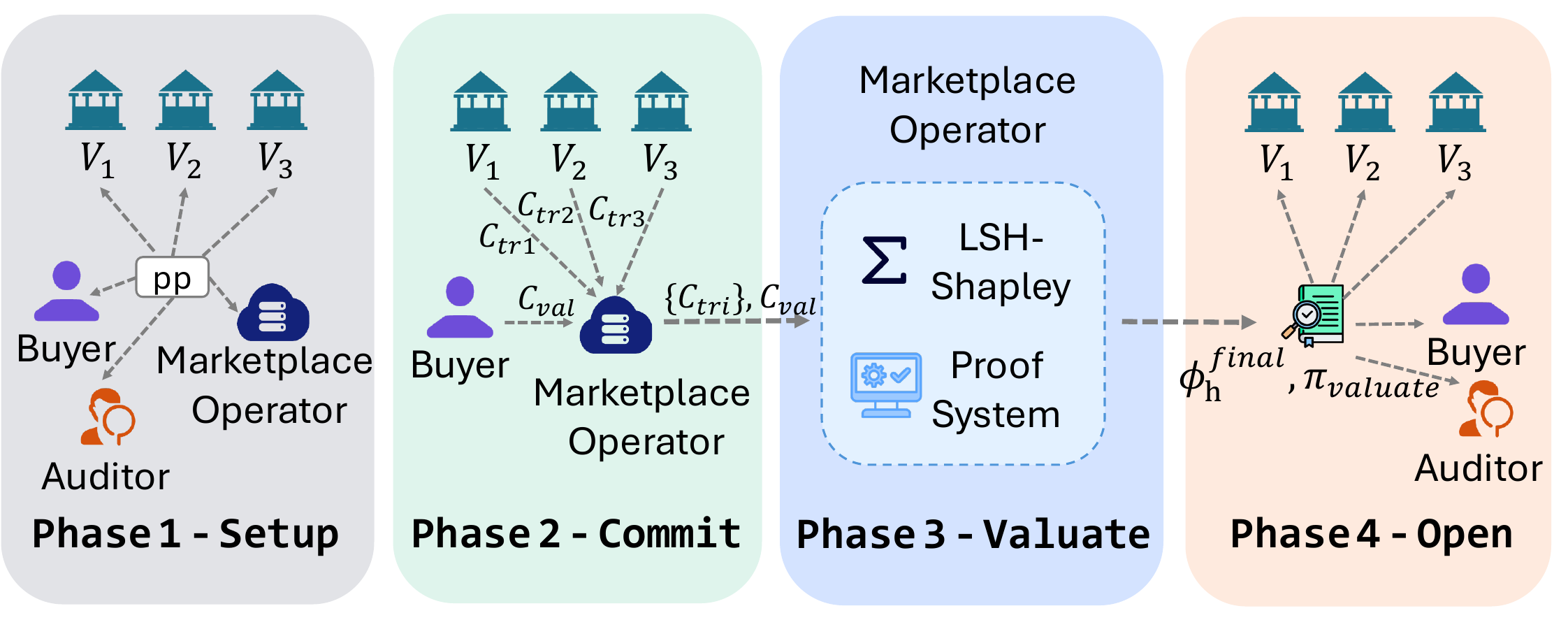}
  \caption{End-to-end architecture of \tool.}
  \label{fig:protocol_architecture}
\end{figure}

\parh{Four protocol phases.}~\tool's valuation pipeline runs in
four sequential phases (Figure~\ref{fig:protocol_architecture}).
\emph{(1) Setup.}~$\Setup(1^\lambda) \to (\pp, \mathsf{vk})$ produces the public parameters and verification key for the polynomial commitment scheme and the Fiat--Shamir-transformed sumcheck backend, once per deployment.
\emph{(2) Commit.}~Providers and the buyer submit their datasets
to $\marketplace$ and publish the corresponding commitments
$\{\commitment{i}\}_{i \in [m]}$ and $\commitment{\mathrm{val}}$.
\emph{(3) Valuate.}~$\marketplace$ ingests the private witnesses (the providers' datasets and the buyer's validation set), computes the per-training-point Shapley scores $\{\phifinal_h\}_{h \in [\Ntrain]}$, and produces a single non-interactive proof $\pi_{\mathrm{valuate}}$ certifying that these scores are consistent with every published commitment.
\emph{(4) Open.}~Each provider $\provider_i$ receives only ${\phifinal_h : h \in \indexset{i}}$ together with a selective opening bound to $\commitment{i}$, while all other providers' scores remain hidden. The opening is accompanied by $\pi_{\mathrm{valuate}}$, which certifies the correctness of the released scores.

\subsection{Threat Model and Security Goals}
\label{sec:system:threat}

\parh{Threat model.}~A valuation transaction touches three assets
that need protection: the provider shards
$\{\subD{i}\}_{i \in [m]}$, the buyer's validation set
$\subD{\mathrm{val}}$, and the integrity of the valutaion scores on those
committed inputs. We assume each role of \S\ref{sec:system:model}
may deviate adversarially against these assets. A dishonest provider $\provider_i$ may craft shards to inflate their score or attempt to recover peers' data from the transcript. A curious buyer may use score signals for inference attacks against shards or "free-ride" by abandoning settlement once the value is revealed. A \emph{malicious operator} $\marketplace$ observes the submitted datasets, executes the prover, and collects a transaction fee, giving it a direct stake in the outcome it computes; it may therefore substitute inputs inconsistent with the commitment, bias the valuation function, or leak per-provider scores to unauthorized parties.

\parh{Security guarantees.}~\tool addresses these threats by leveraging ZKP primitives and their underlying commitment schemes (\S\ref{sec:bg:zkp}). During the commit phase, input integrity is established via binding (the prover cannot later prove against a dataset different from the one each party committed) and hiding (verifiers learn nothing about the committed datasets from the commitments alone). Subsequently, the proof $\pi_{\mathrm{valuate}}$ enforces correctness, guaranteeing that released scores derive strictly from the committed inputs and the agreed valuation function. \tool formally inherits
completeness, knowledge soundness, and zero-knowledge from its
underlying ZKP primtives (\S\ref{sec:bg:zkp}), with specific realizations detailed in \S\ref{sec:method:bucket}.

\subsection{Application Scope}
\label{sec:system:scope}

\tool provides an integrity layer for Shapley-based data valuation for classification tasks in cross-organizational marketplaces. It offers providers, buyers, and auditors public, non-interactive verifiability of released scores against committed inputs, scaling to modern ML workloads. Privacy-focused alternatives address orthogonal goals: DP-TKNN-Shapley~\cite{wang2023privacy} mitigates membership-inference leakage, while MPC distributes computation to hide private inputs. Neither generates a publicly verifiable certificate of correct execution, so they complement \tool rather than substitute for it (detailed in \S\ref{sec:related-work}).

\subsection{The Algorithmic Challenge}
\label{sec:system:challenge}

\parh{The natural baseline.}~The protocol of
\S\ref{sec:system:model} fixes the parties, commitments, and four
phases, but leaves one slot open: which valuation function does
$\marketplace$ certify in $\pi_{\mathrm{valuate}}$? A natural choice is KNN-Shapley~\cite{10.14778/3342263.3342637}. Among the estimators in \S\ref{sec:bg:shapley}, KNN-Shapley is the only one with an $O(\Ntrain \log \Ntrain)$ closed form per validation point, making it the standard tractable proxy in data-valuation literature~\cite{wang2023privacy,wang2023data}. Instantiating \tool\ with KNN-Shapley is thus the most conservative approach: it ensures that the ZK layer carries the verifiability burden without introducing unstudied valuation trade-offs.

\parh{Why the baseline fails ZK encoding.}~Encoding valuation inside $\pi_{\mathrm{valuate}}$ requires certifying every operation against committed inputs. For KNN-Shapley, this imposes two structural costs scaling as $\Ntrain \cdot \Ntest$: (1) computing the full distance matrix, which requires an inner-product commitment per training-validation pair, and (2) ranking these distances, which demands costly non-arithmetic constraints for each validation point. While this $\Ntrain \cdot \Ntest$ complexity is workable in plaintext, ZK certification overhead makes it the dominant proof bottleneck. Tractable ZK valuation therefore requires an algorithm and proof system that decouple from this $\Ntrain \cdot \Ntest$ scaling (Requirement P-4, \S\ref{sec:intro}). We develop both pieces in \S\ref{sec:method}, with further optimizations in \S\ref{sec:optimizations}.

\section{Method}
\label{sec:method}

This section introduces the valuation function that \tool certifies, \emph{LSH-Shapley} (\S\ref{sec:method:plaintext}), together with \zkls{} (\S\ref{sec:method:bucket}), the specialized ZK protocol that produces this certificate.

\subsection{LSH-Shapley}
\label{sec:method:plaintext}

\emph{LSH-Shapley} replaces KNN-Shapley's pairwise-distance signal with bucket co-membership under a locality-sensitive hash. We state the utility in \S\ref{sec:method:plaintext:motivation} and derive its closed-form Shapley value in \S\ref{sec:method:plaintext:spec} (Theorem~\ref{thm:lsh-shapley-closed}).

\subsubsection{Bucket-Based Nearest Neighbor Utility}
\label{sec:method:plaintext:motivation}
Before stating the utility, we fix the hash family and the bucket-count primitives that its closed form (\S\ref{sec:method:plaintext:spec}) is built on.

\textbf{Hash family.}
We instantiate the LSH primitive of \S\ref{sec:bg:lsh} with
SimHash. Each hash bit is the sign of a Gaussian inner product
$\signbit\langle r, x\rangle$ with $r \sim \mathcal{N}(0, I_d)$,
so its underlying metric is angular similarity, the appropriate
notion for $\ell_2$-normalized neural embeddings. Stacking
$\Kdepth$ such bits gives a per-table bucket id
$B_\ell(x) = \hashkey_\ell(x) \in \{0,1\}^{\Kdepth}$ with
expected occupancy $\Ntrain / 2^{\Kdepth}$, and we run
$\Ntables$ independent tables to drive estimator variance down.
Among LSH families, SimHash also has the lowest circuit degree
under our ZK back-end: the inner product $\langle r, x \rangle$
is linear in $x$, and sign extraction reduces to a single lookup.
Alternatives such as $p$-stable hashes~\cite{datar2004locality} (floor gates) and MinHash~\cite{broder1998min}
(set-encoding gates) raise the per-bit circuit degree
substantially.

\textbf{Bucket and counts.}
The \emph{bucket} of validation point $x_t^{\mathrm{te}}$ in
table $\ell$ is the set of training points sharing its bucket id,
\(
\mathcal{B}_\ell^t := \{\,j \in [\Ntrain] :
  B_\ell(x_j^{\mathrm{tr}}) = B_\ell(x_t^{\mathrm{te}})\,\}.
\)
We summarize each $(\ell, t)$ pair by two counting queries on
the committed data: $M(\ell, t) := |\mathcal{B}_\ell^t|$
(bucket size) and
$M^{+}(\ell, t) := |\{j \in \mathcal{B}_\ell^t : y_j = y_t\}|$
(same-label count). Both are simple histograms over bucket ids,
which is what eventually makes the closed form of
\S\ref{sec:method:plaintext:spec} ZK-friendly.

\textbf{Utility function.}
For a coalition $S \subseteq [\Ntrain]$, validation point $t$,
and table $\ell$, we define the LSH-Shapley utility as the
fraction of $S$'s bucket-mates whose label agrees with $y_t$:
\begin{equation}
\label{eq:lsh-utility}
u_\ell^t(S) :=
\begin{cases}
0 & |S \cap \mathcal{B}_\ell^t| = 0, \\[2pt]
\dfrac{\bigl|\{\,j \in S \cap \mathcal{B}_\ell^t : y_j = y_t\,\}\bigr|}
      {\bigl|S \cap \mathcal{B}_\ell^t\bigr|}
  & |S \cap \mathcal{B}_\ell^t| > 0.
\end{cases}
\end{equation}
This form is motivated by three observations.
(i) Non-colliders ($j \notin \mathcal{B}_\ell^t$) are Shapley dummies, so the $\Ntrain$-player game restricts to an $M$-player subgame on the bucket.
(ii) If $\mathcal{B}_\ell^t$ equals a $K$-nearest-neighbor set, $u_\ell^t(S)$ reduces to the fraction-of-correct-votes utility of~\cite{10.14778/3342263.3342637}, preserving KNN-Shapley semantics.
(iii) For any coalition $S$, the utility $u_\ell^t(S)$ depends on the data only through two integers, $|S \cap \mathcal{B}_\ell^t|$ and $|\{j \in S \cap \mathcal{B}_\ell^t : y_j = y_t\}|$. This dependence on aggregate bucket counts yields the closed-form Shapley value of \Thm~\ref{thm:lsh-shapley-closed}, expressed solely through $M(\ell, t)$ and $M^{+}(\ell, t)$, letting LSH-Shapley avoid the sorting and pairwise-distance proofs that KNN-Shapley would otherwise require in ZKP.

\subsubsection{Closed-Form Shapley Values}
\label{sec:method:plaintext:spec}

With the utility $u_\ell^t$ fixed by \eqref{eq:lsh-utility}, we now derive its closed-form Shapley value. This closed form is the structural enabler of \tool: it reduces the exponential coalition sum to a constant-time expression in two histogram counts, and it determines the algebraic shape of the ZKP protocol of \S\ref{sec:method:bucket} (the protocol design is the paper's main contribution). For brevity, let $\harm{M} := \sum_{j=1}^{M} 1/j$ denote the $M$-th harmonic number, and let $M$ and $M^{+}$ stand for the bucket counts $M(\ell, t)$ and $M^{+}(\ell, t)$ when no ambiguity arises. With this notation in place, the Shapley value of $u_\ell^t$ admits the following closed form.

\begin{theorem}[LSH-Shapley closed form]
\label{thm:lsh-shapley-closed}
Let $\phi^{\pm}(\ell, i, t)$ denote the Shapley value of training point $i$ under $u_\ell^t$, where the $+$ branch applies when the labels match ($y_i = y_t$) and the $-$ branch applies when they disagree. The value depends on the data only through $(M, M^{+}, \harm{M})$. For non-colliders ($i \notin \mathcal{B}_\ell^t$), $\phi^{\pm} = 0$; for singleton buckets ($M = 1$), $\phi^{\pm} = \mathbf{1}[y_i = y_t]$; for $i \in \mathcal{B}_\ell^t$ with $M \geq 2$,
\begin{equation}
\label{eq:phi-pm}
\phi^{\pm} = \begin{cases}
\dfrac{\harm{M}}{M} - \dfrac{(M^{+}-1)(\harm{M}-1)}{M(M-1)} & y_i = y_t, \\[6pt]
-\,\dfrac{M^{+}(\harm{M}-1)}{M(M-1)} & y_i \neq y_t.
\end{cases}
\end{equation}
\end{theorem}

The proof, the verification of all four Shapley axioms, and the
sign / monotonicity properties are deferred to
\appref{app:closed-form}.

\textbf{Released score per training point.} Combining Theorem~\ref{thm:lsh-shapley-closed} with the linearity of the Shapley value in the utility, the score \tool{} releases for training point $i$ is the average of the per-table, per-validation contributions,
\[
\svavg_i \;:=\; \frac{1}{\Ntables \cdot \Ntest} \sum_{\ell, t} \phi^{\pm}(\ell, i, t),
\]
and we also work with the unnormalized form $\svsum_i := \Ntables \cdot \Ntest \cdot \svavg_i$, which is what the table-aggregation sumcheck of \S\ref{sec:method:bucket} certifies in the verified circuit.

\textbf{Why averaging, not OR-amplification.} \tool{} aggregates the $\Ntables$ tables by mean rather than the OR-amplification rule of standard LSH nearest-neighbor search (\S\ref{sec:bg:lsh}). \emph{Statistically}, the per-table games are i.i.d. random partitions, so the table-mean is unbiased for its expectation over the random hash projections (concentration rate below). \emph{Algorithmically}, Shapley linearity in the utility decomposes the table-mean utility's Shapley value into $\Ntables$ independent per-table evaluations of Theorem~\ref{thm:lsh-shapley-closed}, with the released score their arithmetic mean. OR-amplification breaks this: it combines per-table memberships by a non-linear Boolean OR rather than a sum, so Shapley linearity no longer factors across tables and the per-table closed form does not transfer. This Boolean OR is also more expensive to certify in the ZK backend than an arithmetic mean.

\textbf{Algorithm.} Algorithm~\ref{alg:lsh-shapley} presents LSH-Shapley as four phases that map one-for-one onto the verification obligations of \S\ref{sec:method:bucket:sumchecks}: Phase~1 hashes every training and validation point, Phase~2 queries the bucket counts $(M, M^{+})$ off those hashes, Phase~3 evaluates $\phi^{\pm}$ via \eqref{eq:phi-pm} against a precomputed harmonic table, and Phase~4 aggregates per-sample scores across the $\Ntables \cdot \Ntest$ table-validation pairs.

\begin{algorithm}[t]
\caption{LSH-Shapley}
\label{alg:lsh-shapley}
\begin{algorithmic}[1]
\Require Training set
$\{(x_j^{\mathrm{tr}}, y_j)\}_{j=1}^{\Ntrain}$,
validation set $\{(x_t^{\mathrm{te}}, y_t)\}_{t=1}^{\Ntest}$,
hash tables $\Ntables$, hash depth $\Kdepth$, public Gaussian
projection vectors
$r_\ell^{(k)} \sim \mathcal{N}(0, I_d)$ for
$\ell \in [\Ntables]$, $k \in [\Kdepth]$.
\Ensure Per-sample scores $\{\svavg_i\}_{i=1}^{\Ntrain}$.
\State \textbf{Phase~1 (Hash).} For each
$\ell \in [\Ntables]$ and each input
$x \in \{x_j^{\mathrm{tr}}\}_{j} \cup \{x_t^{\mathrm{te}}\}_{t}$,
compute the bucket id
$B_\ell(x) \gets
\bigl(\signbit\langle r_\ell^{(k)}, x\rangle\bigr)_{k < \Kdepth}$.
\State \textbf{Phase~2 (Bucket statistics).} For each
$(\ell, t)$, form
$\mathcal{B}_\ell^t = \{j : B_\ell(x_j^{\mathrm{tr}}) =
  B_\ell(x_t^{\mathrm{te}})\}$
and record the counts $M(\ell, t)$ and $M^{+}(\ell, t)$.
\State \textbf{Phase~3 (Per-bucket weight).} For each
$(\ell, t, i)$, evaluate $\phi^{\pm}(\ell, i, t)$ via
\eqref{eq:phi-pm} using a precomputed harmonic table
$\{\harm{m}\}_{m=1}^{\Ntrain}$.
\State \textbf{Phase~4 (Aggregate).} For each
$i \in [\Ntrain]$, set
$\svavg_i \gets (\Ntables \cdot \Ntest)^{-1}
  \sum_{\ell, t} \phi^{\pm}(\ell, i, t)$.
\State \Return $\{\svavg_i\}_{i=1}^{\Ntrain}$.
\end{algorithmic}
\end{algorithm}

\textbf{Complexity.} A single evaluation of \eqref{eq:phi-pm} on the inputs $(M, M^{+}, \harm{M})$ runs in $O(1)$ time. There are $\Ntables \cdot \Ntest \cdot \Ntrain$ triples $(\ell, i, t)$ in total, so computing all of $\{\phi^{\pm}(\ell, i, t)\}$, together with the bucket histograms they depend on, costs $O(\Ntables \cdot \Ntest \cdot \Ntrain)$ overall. The feature dimension $d$ enters only once, through the hashing step (Phase~1 of Algorithm~\ref{alg:lsh-shapley}).

\textbf{Convergence.} Theorem~\ref{thm:lsh-shapley-closed} is exact for each per-table game $u_\ell^t$, so the only randomness in $\svavg_i$ is over the hash projections $r_\ell^{(k)}$. We therefore view $\svavg_i$ as a Monte~Carlo estimator of its expectation $\mathbb{E}_{r}[\svavg_i]$ over the random hash family. Since the $\Ntables$ tables are i.i.d. random partitions, this estimator is unbiased and concentrates at rate $O(1/\sqrt{\Ntables})$ around its expectation; the formal statement and the rank-stability bound that drives our parameter selection are in Theorem~1 of \appref{app:thm1-proof}.

\subsection{ZKP Protocol for LSH-Shapley}
\label{sec:method:bucket}

We now describe \zkls{}, the ZKP protocol that certifies Theorem~\ref{thm:lsh-shapley-closed}'s closed-form score on committed training and validation inputs.

\subsubsection{Witness Design: Why Bucket Histograms}
\label{sec:method:bucket:witnesses}

A key design choice in \zkls{} is the witness granularity. By \Thm~\ref{thm:lsh-shapley-closed}, $\phi^{\pm}$ depends on an $(\ell, i, t)$ triple solely through the bucket counts $(M(\ell, t), M^{+}(\ell, t))$ and the harmonic value $\harm{M}$. The witness must deliver these efficiently. We compare a natural per-pair baseline against \zkls{}'s bucket-level encoding.

\textbf{Naive encoding (\zklsnaive{}).} A direct per-pair approach commits a binary tensor $\member(\ell, t, i) = \mathbf{1}[B_\ell(x_i^{\mathrm{tr}}) = B_\ell(x_t^{\mathrm{te}})]$ indicating bucket co-membership. This requires full-tensor consistency sumchecks to recover $M$ and $M^{+}$. This incurs two avoidable costs. First, witness size scales with the validation-set size $\Ntest$, unnecessarily tracking per-pair interactions when only bucket-level counts are needed. Second, verifying each bit forces sumchecks to re-derive bucket IDs from $\Kdepth$ sign bits, pushing the per-round polynomial degree to $\Kdepth + O(1)$. Consequently, both prover work and proof size scale poorly as $\mathcal{O}(\Ntest \cdot \Ntables \cdot \Ntrain)$.

\textbf{Bucket histograms.} To avoid overcommitting, \zkls{} directly commits bucket statistics. The base histogram $\cnt[\ell, b]$ and its per-class refinements $\cnttr[\ell, b, c]$ and $\cntte[\ell, b, c]$ count points hashing to each bucket. Spanning $\mathcal{O}(\Ntables \cdot \Nbuckets \cdot \Nclass)$ entries, this bucket-side witness drops the $\Ntest$ axis entirely. To recover per-sample statistics, we define three \emph{materialized lookup oracles} indexing the pre-resolved count at training point $i$'s bucket:
\[
\begin{aligned}
\mhat[\ell, i]  &\;:=\; \cnt\bigl[\ell,\, B_\ell(x_i^{\mathrm{tr}})\bigr], \\
\mchat[\ell, i] &\;:=\; \cnttr\bigl[\ell,\, B_\ell(x_i^{\mathrm{tr}}),\, y_i\bigr], \\
\tchat[\ell, i] &\;:=\; \cntte\bigl[\ell,\, B_\ell(x_i^{\mathrm{tr}}),\, y_i\bigr].
\end{aligned}
\]
The reason for committing $\mhat, \mchat, \tchat$ separately rather than letting the weight sumcheck look up $\cnt[\ell, B_\ell(x_i^{\mathrm{tr}})]$ on the fly is to keep the per-round polynomial degree low. An on-the-fly lookup would entangle the histogram data with the bucket-routing bits, drastically increasing the computational overhead of the proof. Pre-computing these arrays allows the main proof to query them at minimal cost, requiring only a small, one-time auxiliary check to ensure they correctly match the base histograms. The harmonic value $\harm{M}$ also enters the closed form, but it is a deterministic public function of $M$; once the histograms establish $M$, $\harm{M}$ is implicitly verified without requiring an extra step.

\textbf{Commitment size.} The full protocol commits five families of oracles. The bucket-side family ($\cnt, \cnttr, \cntte$ and the per-sample lookups $\mhat, \mchat, \tchat$, all introduced above) takes $O(\Ntables \cdot \Ntrain + \Ntables \cdot \Nbuckets \cdot \Nclass)$, independent of $\Ntest$. The hashing-side family (the feature MLEs $\MLExtrain, \MLExtest$ together with the sign and dot-product oracles that bind the LSH evaluation to those features) adds $\Ntest \cdot d$ for the validation features (where $d$ denotes the feature dimension). The per-class label indicators $y_{\mathrm{ind}}, y_{\mathrm{test\_ind}}$ that resolve the same-label test inside the closed form add $\Nclass \cdot \Ntest$. The weight-identity advice ($W, e, D_{\mathrm{inv}}, \mathsf{wt\_num}$) that turns the closed-form division and edge-case split into low-degree polynomial identities, together with the aggregate oracles $\svsum, \svavg$, contribute $O(\Ntables \cdot \Ntrain)$. Table~\ref{tab:plaintext-to-zk-mapping} lists every committed oracle in dataflow order, paired with the protocol step that consumes it.

\begin{table}[t]
\small
\caption{Committed-oracle inventory for \zkls{}, in dataflow order. Each plaintext object of \S\ref{sec:method:plaintext} is mapped to the committed oracle and the protocol step that consumes it.}
\label{tab:plaintext-to-zk-mapping}
\centering
\begin{tabular}{@{}p{2.5cm}p{3.4cm}p{2.0cm}@{}}
\toprule
\textbf{Plaintext object} & \textbf{Committed oracle} & \textbf{Verified by} \\
\midrule
Features (train / test) & $\MLExtrain, \MLExtest$ & Hash-binding \\
Hash bits & $\signbit, \mathsf{abs\_dp}, \dotprod$ & Hash-binding \\
Per-bucket histogram & $\cnt$ & Histogram \\
Per-class histograms & $\cnttr, \cntte$ & Histogram \\
Label indicators & $y_{\mathrm{ind}}, y_{\mathrm{test\_ind}}$ & Histogram \\
Per-sample lookups & $\mhat, \mchat, \tchat$ & Lookup \\
Shapley weight & $W, e, D_{\mathrm{inv}}, \mathsf{wt\_num}$ & Weight \\
Aggregates & $\svsum, \svavg$ & Weight \\
Harmonic value & $\harm{M}$ (not committed) & Reconstruction \\
\bottomrule
\end{tabular}
\end{table}

\subsubsection{Protocol Construction}
\label{sec:method:bucket:sumchecks}

Given the witnesses of \S\ref{sec:method:bucket:witnesses}, the verifier discharges four obligations Algorithm~\ref{alg:lsh-shapley} imposes (honest hashing, histogram faithfulness, closed-form weight, aggregation) plus the marketplace requirement (commit-binding to providers and disclosure to buyer). These factor into a hashing module upstream and four chained verification layers downstream, each binding the next adjacent commitment level via a small group of sumcheck modules, summarized at a glance in Algorithm~\ref{alg:lsh-bucket}.

\textbf{Hashing module.} The hashing module sits upstream of the four layers and discharges the honest-hashing obligation by tying the sign and dot-product oracles to the committed feature MLEs. We utilize existing sumcheck protocols for sign, Boolean, and absolute dot-product operations~\cite{zhang2025fairzk}, which verify that $\signbit\langle r_\ell^{(k)}, x \rangle$ is Boolean and that the absolute dot products are correctly computed. This architecture successfully isolates the feature dimension to the input layer; downstream sumchecks process only verified bucket assignments, never raw features.

\textbf{Layer~1: histogram consistency.} With the bucket assignments verified upstream, Layer~1 binds the committed histograms to those assignments, discharging the histogram-faithfulness obligation (Phase~2 of Algorithm~\ref{alg:lsh-shapley}). The training-histogram consistency sumcheck binds $\cnt$ via the identity
\[
\sum_{(\ell, i) \in \Bool^{\log\Ntables + \log\Ntrain}} \mathrm{eq}_\Kdepth\bigl(B_\ell(x_i^{\mathrm{tr}}),\, b\bigr) \;=\; \cnt[\ell, b].
\]
Two analogous sumchecks bind $\cnttr$ and $\cntte$ on the per-class refinements (the per-class training-histogram consistency sumcheck and the per-class validation-histogram consistency sumcheck); per-module target sums and round-degrees are tabulated in \appref{app:module-table}. Round-degrees are constant in $\Kdepth$ ($\Kdepth + 1$ to $\Kdepth + 2$); this layer dominates prover time.

\textbf{Layer~2: lookup consistency.} With the histograms now bound to the bucket assignments, Layer~2 binds the per-sample materialized lookups to those histograms, completing the second half of Phase~2's faithfulness obligation. The bucket-size lookup sumcheck binds $\mhat$ via
\[
\mhat[\ell, i] \;=\; \cnt\bigl[\ell,\, B_\ell(x_i^{\mathrm{tr}})\bigr],
\]
verified through an $\mathrm{eq}_\Kdepth$-based identity sumcheck on $(\Ntables, \Ntrain)$ paired with a degree-3 squared-count auxiliary on $(\Ntables, \Nbuckets)$. Two analogous lookup sumchecks (the same-class lookup and the validation-class lookup, each main and squared-count auxiliary) bind $\mchat$ and $\tchat$ to their respective histograms; full identities in \appref{app:module-table}.

\textbf{Layer~3: weight identity and aggregation.} With per-sample lookups now bound to the histograms, Layer~3 plugs them into the closed form: it verifies that the committed weight $W$ satisfies~\eqref{eq:phi-pm} on every training point and that the published $\svavg$ is the correct aggregate (Phases~3 and~4 of Algorithm~\ref{alg:lsh-shapley}). The edge-case weight identity binds $W$ on the $M \in \{0, 1\}$ region:
\[
W \;=\; \mhat[\ell, i] \cdot \Bigl(\sum_c y_{\mathrm{ind},c}[i] \cdot \tchat[\ell, i]\Bigr) \quad \text{on } e = 1.
\]
The interior weight identity ($M \geq 2$) clears the denominator $M(M-1)$ via a committed inverse $D_{\mathrm{inv}}$ and a numerator witness $\mathsf{wt\_num}$; auxiliary checks discharge the boolean of the edge flag $e$ and the denominator-inverse structural identities. The table-aggregation sumcheck collapses $W$ across $\Ntables$ tables into $\svsum$, and a normalization sumcheck~\cite{zhang2025fairzk} divides by $\Ntables \cdot \Ntest$ to obtain $\svavg$.

\textbf{Layer~4: input/output reconstruction.} Layers~1--3 produce a single committed $\svavg$ from honest LSH-Shapley evaluation, but in the marketplace setting that is not enough: the protocol must additionally bind $\svavg$ to the providers who actually contributed and let each buyer privately verify only their own scores. Layer~4 closes the binding chain on both sides, tying $\svavg$ back, on the input side, to per-provider data commitments $\cm{\subD{i}}$, and forward, on the output side, to per-provider score-slice commitments $\cm{\subphi{i}}$. Two structurally identical reconstruction sumchecks discharge the two halves. For each provider $i$ over the disjoint contiguous slice $\indexset{i}$ with $h = (h_{\mathrm{low}}, h_{\mathrm{high}})$, the output-side reconstruction proves
\[
\MLE{\svavg}(h) \;=\; \sum_{i=1}^{m} \eq(h_{\mathrm{high}}, \mathrm{bin}(i)) \cdot \subphi{i}(h_{\mathrm{low}}),
\]
reducing to one PCS opening of $\svavg$ at $\tau$ plus $m$ openings of $\subphi{i}$ at $\ptlow$; the input-side reconstruction is symmetric, binding $\MLExtrain$ to $\{\cm{\subD{i}}\}$. Together the two reconstructions yield the binding chain
\[
\begin{aligned}
&D_i \to \cm{\subD{i}} \to \cm{\MLExtrain} \to \cdots \\
&\qquad \to \cm{\svavg} \to \cm{\subphi{i}} \to \svavg_{|_{\indexset{i}}}.
\end{aligned}
\]
Each provider then runs three local checks (input binding, output binding, transcript correctness) on the per-provider slice they verify, deferred to \appref{app:op7-full}.

\textbf{Optional layer for high-dimensional inputs.} When the raw feature dimension is too large to commit directly under the polynomial-commitment scheme, we adopt a verifiable dimensionality-reduction (PCA) sumcheck from prior work~\cite{zhang2025fairzk} as a preprocessing step that projects the features into a lower-dimensional space before hashing. The construction commits the necessary auxiliary witnesses and discharges the projection identity through additional sumchecks; we use it as a black box and defer details to \appref{app:pseudocode}.

\begin{algorithm}[t]
\caption{\zkls{} (high-level). \ifextended The full per-module protocol is Algorithm~\ref{alg:lsh-bucket-full} in Appendix~\ref{app:pseudocode}.\else The full per-module protocol is given in Appendix~\cite{extendedversion}.\fi}
\label{alg:lsh-bucket}
\begin{algorithmic}[1]
\Statex \textbf{Public input:} public parameters $\pp$, projection vectors $\{r_\ell^{(k)}\}$, per-provider data commitments $\{\cm{\subD{i}}\}$, validation-set commitment $\cm{\subD{\mathrm{val}}}$, hash parameters $(\Ntables, \Kdepth)$.
\Statex \textbf{Witness:} per-provider data shards $\{\subD{i}\}$, validation set $\subD{\mathrm{val}}$; bucket/weight/label oracles (see Table~\ref{tab:plaintext-to-zk-mapping}).
\Statex \textbf{Output:} proof $\pi$ certifying $\svavg = \mathsf{LSH\text{-}Shapley}(\subD{}, \subD{\mathrm{val}})$.
\Statex
\State \textbf{Commit witnesses.} Prover commits to feature MLEs, hashing/weight oracles, bucket histograms, label indicators, and aggregates (plus PCA witnesses if $d > \text{PCS budget}$).
\State \textbf{Hash binding.} Run the hash-binding sumchecks that tie sign and dot-product oracles to the committed feature MLEs; if PCA is enabled, additionally run the ZK-PCA sumchecks.
\State \textbf{Histogram and lookup consistency.} Run sumchecks binding bucket counts $(\cnt, \cnttr, \cntte)$ to the bucketassignments and lookups $(\mhat, \mchat, \tchat)$ to those histograms.
\State \textbf{Weight and aggregation.} Run the weight-identity sumchecks binding $W$ to the closed form~\eqref{eq:phi-pm} on both the edge and interior cases, then the table-aggregation and normalization sumchecks producing $\svsum$ and $\svavg$.
\State \textbf{Reconstruction and opening.} Run the input-side and output-side reconstruction sumchecks binding $\MLExtrain$ and $\svavg$ to per-provider commitments, then send PCS openings at every sumcheck challenge point.
\end{algorithmic}
\end{algorithm}

\subsubsection{Security and Complexity}
\label{sec:method:bucket:soundness}

This subsection states the security guarantees and prover-cost profile of \zkls{}, with formal proofs and per-module breakdowns deferred to \appref{app:soundness-chain}.

\paragraph{Security.} \zkls{} is a non-interactive zero-knowledge argument of knowledge for the closed-form score relation: it satisfies completeness (an honest prover always convinces the verifier), knowledge soundness (the verifier accepts only when the prover knows a witness consistent with the closed-form score formula~\eqref{eq:phi-pm}), and zero-knowledge (the proof reveals nothing beyond the published $\svavg$). 

\paragraph{Complexity.} \phantomsection\label{sec:method:bucket:complexity}
The hashing module incurs a one-shot input-binding cost of $O(\Ntables \cdot \Kdepth \cdot (\Ntrain + \Ntest) \cdot d)$ and is compute-dominant due to sign-check overhead, with $d$ and $\Ntest$ entering only here and dropping out of all downstream layers. Subsequent Shapley layers are led by histogram-consistency sumchecks at $O(\Ntables \cdot \Ntrain \cdot \mathrm{poly}(\Kdepth))$, with lookup, weight, and aggregation costs lower order; Layers~2 and~3 instantiate one sumcheck instance per class, contributing a multiplicative $\Nclass$ factor on those layers' transcript. Proof size therefore splits into a transcript component linear in $\Nclass$ on Layers~2 and~3, plus a PCS-opening component governed by the total length of the $\Nclass$-indexed tensor and batched across classes (\S\ref{sec:opt:super-oracle}).

\section{Proof-System Optimizations}
\label{sec:optimizations}

The \zkls{} protocol of \S\ref{sec:method} is correct but
unoptimized along two axes: proof size and prover time. This
section presents two optimizations that target them: super-oracle
batching (\S\ref{sec:opt:super-oracle}) and sparsity-aware
sumchecks (\S\ref{sec:opt:sparsity}).

\subsection{Super-Oracle Batching for Proof Size}
\label{sec:opt:super-oracle}

\parh{The opening-count bottleneck.} Under our PCS
(\S\ref{sec:bg:zkp}), if $Q$
denotes the number of PCS openings and $c_{\mathrm{fix}}$ the
per-opening fixed-overhead bytes (commitment header,
authentication path, transcript material), proof size scales as
\[
  |\pi_{\mathrm{PCS}}| \;=\; Q \cdot \bigl(c_{\mathrm{fix}} + \sqrt{\Nbuckets}\bigr),
\]
so the opening count $Q$ is the dominant lever, and each opening
also costs the prover one PCS open phase. The witnesses
committed in \S\ref{sec:method:bucket:witnesses} are indexed
tensors, each queried by the sumcheck that consumes it (its
\emph{parent sumcheck}) at one or more evaluation points. Most
parents query a single point and contribute a single opening;
the cost-bearing exceptions are tensors whose parent issues
\emph{multiple} queries at one shared upstream point, differing
only along a single index axis. We call such an indexed slice
group a \emph{family}, and three of them dominate $Q$ in
\zkls{}:
\begin{itemize}[leftmargin=1.5em,topsep=2pt,itemsep=1pt]
  \item the \emph{sign-bit family} $\signbit[\ell, k, \cdot]$,
  driven by the dot-product sign sumchecks of
  \S\ref{sec:method:bucket:sumchecks}, which probe $\Kdepth$
  points differing only in the hash-function index $k$ on each of
  the training-side and test-side preambles;
  \item the \emph{histogram family} $\cnttr[\ell, b, c]$, driven
  by the per-class training-histogram sumcheck together with the
  bucket-size lookup, which together probe $2\Nclass + 2$ points
  differing only in the class coordinate $c$;
  \item the \emph{lookup family}
  $\{\mchat[\ell, i, c], \tchat[\ell, i, c]\}$, driven by the
  same-class and validation-class lookup sumchecks, which each
  probe $\Nclass$ points differing only in $c$.
\end{itemize}
In every case the queries share a common upstream point
$x^* \in \EF^{n}$ (the random suffix of verifier challenges
already fixed by the parent sumcheck) and differ only in the
index variable that one prepends to $x^*$. This is exactly the
structure batching exploits.

\parh{Construction.} Let $f_0, f_1, \ldots, f_{Q-1} : \Bool^{n} \to \Fld$
be $Q$ multilinear \emph{sub-oracles}, each obtained by fixing one
index axis of a parent tensor (e.g.,
$f_c(\ell, b) := \cnttr[\ell, b, c]$ for $c \in \{0, \ldots, \Nclass - 1\}$).
Suppose the upstream sumcheck has fixed a common evaluation point
$x^{\star} \in \EF^{n}$ and the verifier needs the $Q$ values
$f_0(x^{\star}), \ldots, f_{Q-1}(x^{\star})$ to close its claim;
without batching, each $f_j$ carries its own PCS commitment and
opening at $x^{\star}$, costing
$Q \cdot (c_{\mathrm{fix}} + \sqrt{\Nbuckets'})$ proof bytes and
$Q$ separate open-phase invocations on the prover. We collapse the
$Q$ commitments to one by padding $Q$ to a power of two and
reshaping the family $\{f_j\}$ into a single multilinear extension
on $\Bool^{\log Q + n}$ that treats the index $j$ as an extra
variable:
\[
  F(j, x) \;:=\; f_j(x), \quad
  \text{extended multilinearly to } (j, x) \in \EF^{\log Q + n}.
\]
We call $F$ the \emph{super-oracle} of the family $\{f_j\}$. The
prover commits $F$ once. To recover the $Q$ values
$\{f_j(x^{\star})\}$, the verifier samples one Fiat--Shamir
challenge $r \in \EF^{\log Q}$, asks for the single PCS opening
$F(r, x^{\star})$, and accepts the prover-supplied values
$v_0, \ldots, v_{Q-1}$ iff
\begin{equation}\label{eq:super-oracle-rlc}
  F(r, x^{\star}) \;=\; \sum_{j \in \Bool^{\log Q}}
  \eq(r, j) \cdot v_j.
\end{equation}
By Schwartz--Zippel over $\EF$, agreement at random $r$ binds
$v_j = f_j(x^{\star})$ for every $j$ except with probability at
most $\log Q / |\EF|$. The verifier then plugs the recovered
$\{v_j\}$ into the parent sumcheck identity exactly as in the
unbatched protocol; nothing downstream changes.

\parh{Application to the three families.} The sign-bit,
histogram, and lookup families introduced above all satisfy the
setup of~\eqref{eq:super-oracle-rlc}: each is a slice of its parent
tensor along a single integer-encoded axis ($k$ for $\signbit$; $c$
for $\mchat$, $\tchat$, and $\cnttr$), and within each family the
$Q$ queries share a common upstream $x^{\star}$ already fixed by
the parent sumcheck. The same construction therefore applies
verbatim to all three. We pack them into \emph{three separate}
super-oracles rather than one: although the template is shared,
the parent tensors have different shapes and queries in different
families land at different $x^{\star}$, so merging the commitments
would leave the post-batching opening count unchanged.

\parh{Savings.} For one super-oracle of $Q$ sub-oracles, the $Q$
separate openings of the unbatched template collapse to one in the
PCS-opening component of the proof:
\[
  \text{before:}\;\; Q \cdot \bigl(c_{\mathrm{fix}} + \sqrt{\Nbuckets'}\bigr)
  \quad\longrightarrow\quad
  \text{after:}\;\; c_{\mathrm{fix}} + \sqrt{Q \cdot \Nbuckets'},
\]
and the prover runs one open phase instead of $Q$, cutting prover
time in proportion. The savings apply to PCS-opening bytes only;
the parent sumchecks that consume the sub-oracles are unaffected,
so any sumcheck transcript whose instance count scales with the
batched index (e.g., per-class lookup and weight-identity sumchecks
over the class axis $c$) retains its linear-in-$Q$ transcript
contribution. Batching is profitable only when $Q$ is large
enough to amortize $c_{\mathrm{fix}}$ and when $\sqrt{\Nbuckets'}$
exceeds the per-opening authentication-path bytes
(see \S\ref{sec:bg:zkp}); families that fall short of either
threshold are left unbatched. We report the per-family breakdown
in \S\ref{sec:eval:ablation}.

\parh{Soundness.} Replacing $Q$ separate openings with a single
packed opening preserves soundness: the random index challenge
$r$ binds the prover to consistent sub-oracle values
via~\eqref{eq:super-oracle-rlc}, and the verifier's downstream
check is unchanged. \Appref{app:super-oracle-soundness} gives the full proof.

\subsection{Sparsity-Aware Sumchecks for Prover Time}
\label{sec:opt:sparsity}

\parh{The dense-grid bottleneck.} The complexity analysis
of \S\ref{sec:method:bucket:complexity} identifies the three
histogram sumchecks of Layer~1 (training-histogram, per-class
training-histogram, per-class validation-histogram) as the
prover-time bottleneck. Each runs $\Kdepth$ rounds, and in every
round the prover walks the entire bucket grid of size
$\Nbuckets = 2^{\Kdepth}$ to assemble its univariate round
message: roughly $\Ntables \cdot \Nbuckets$ field multiplications
per round per sumcheck, repeated across all three sumchecks. The
optimization below brings this cost down to a function of the
data size rather than of the bucket-grid size.

\parh{Why most of that work is wasted.} Every multilinear
extension committed in a sumcheck lives on a power-of-two cube of
size $2^n$, and the $\Kdepth$-round histogram sumchecks of
\S\ref{sec:method:bucket:sumchecks} are no exception. What is
special here is how little of that cube the data actually
occupies: LSH places the $\Ntrain$ training points into at most
$\Ntrain$ distinct cells of $\cnt[\ell, b]$ out of
$\Nbuckets = 2^{\Kdepth}$, and class-stratified factors
($\cnttr[\ell, b, c]$, $\cntte[\ell, b]$, and the per-class label
and lookup slices) inherit the same bucket sparsity plus a
further partition across $\Nclass$ classes. The remaining cells
are zero by construction and contribute nothing to any round
message, yet the dense walk multiplies and accumulates them
anyway. The same observation underpins prior sparsity-aware sumcheck
designs for ZK machine
learning~\cite{li2024sparsity,qu2025zkgpt}, where the zeros
come from rounding arbitrary matrix dimensions up to powers of
two; in our setting the zeros come from hashing and class slicing
instead, but the prover-side mechanism is the same.

\parh{Sparse-aware sumcheck path.} We expose the skip in the
prover by attaching a \emph{support mask}
$\supp(f) \subseteq \Bool^{\log \Nbuckets}$ to each sparse factor
$f \in \{\cnttr, \cntte, \text{per-class label and lookup slices}\}$:
the index set of cells the prover knows are nonzero. To assemble
round $i$'s message, the prover (i) intersects the masks of the
factors that appear in each product term, (ii) iterates the round
message's outer loop only over that intersection, and (iii) folds
each mask in place once the round advances, keeping it one
variable ahead of the partially-folded table. By construction, the
round message emitted on the sparse path equals the dense-path
message bit for bit, so verifier-side soundness transfers
verbatim. A factor whose post-fold density exceeds a fixed
crossover threshold (set so that mask bookkeeping no longer pays
for the zeros it skips) drops back to the dense path for its
remaining rounds; on workloads with high bucket occupancy
$\lambda = \Ntrain / \Nbuckets$, $\cnttr$ itself crosses early and
the savings come only from the sparser per-class factors. Full
pseudocode and the support-superset invariant are in \appref{app:sparsity-aware}.

\parh{Savings.} Unlike super-oracle batching, whose savings
depend only on $Q$, the sparsity-aware path has no fixed
complexity: every multiplication it skips corresponds to a zero
cell of the underlying tensor, and the fraction of zero cells is
set by the workload (bucket occupancy
$\lambda = \Ntrain / \Nbuckets$, class count $\Nclass$, and how
unevenly LSH distributes points across buckets and classes), not
by the protocol. On workloads where the histogram itself stays
dense, the gain comes only from the sparser per-class slices; on
workloads with low bucket occupancy, the histogram contributes
too. We report the per-dataset breakdown in
\S\ref{sec:eval:ablation}.

\parh{Soundness.} The cells the prover skips are exactly the zero
cells of the underlying tensor, and each such cell contributes $0$
to the round-message sum; the message the prover emits on the
sparse path is therefore identical to the dense-path message by
direct equality. Soundness inherits from the unmodified \zkls{}
protocol. \Appref{app:sparsity-aware} gives the full argument.

\providecommand{\TBDPlaceholder}{\textcolor{red}{TBD}}
\providecommand{\TBD}{\TBDPlaceholder}
\providecommand{\knnzk}{\mbox{KNN-Shapley-ZK}\xspace}
\providecommand{\gshapley}{\mbox{G-Shapley}\xspace}

\section{Evaluation}
\label{sec:evaluation}
We aim to answer three research questions (RQs) about \tool.
\begin{itemize}[leftmargin=1.5em,itemsep=2pt]
    \item \textbf{RQ1: Capability.} Does \tool deliver verifiable
    Shapley attributions on real marketplace data while preserving the
    underlying valuation quality?
    \item \textbf{RQ2: Scalability.} How does each cost dimension
    scale as the workload is pushed to modern ML scale along
    $\Ntrain$, $d$, $\Nclass$, and $\Ntest$?
    \item \textbf{RQ3: Ablation.} How much does each component of the
    construction contribute to the
    end-to-end cost?
\end{itemize}

\subsection{Experimental Setup}
\label{sec:eval:setup}

\parh{Datasets.}~We conduct our experiments on the 12-dataset
benchmark established in prior data-valuation
work~\cite{wang2023privacy,liu2026bridging}, which covers the two
data regimes a production marketplace typically trades in: 10 OpenML tabular benchmarks
covering raw feature dimensions $d \in [5, 170]$, and 2 ViT
embeddings on \texttt{mnist} and \texttt{cifar10} at $d = 768$,
standing in for marketplaces that trade pre-processed representations
of rich media. The scalability experiments in
\S\ref{sec:eval:scalability} and the ablation in
\S\ref{sec:eval:ablation} switch to synthetic Gaussian workloads with
features drawn i.i.d.\ from $\mathcal{N}(0, I_{d})$ and class
labels assigned by nearest-cosine to $\Nclass$ random unit anchors,
so we can vary one axis at a time. Pre-processing and synthetic-data
generation details are given in \appref{app:datasets}.

\parh{Baselines.}~We compare \tool against six baselines, organized
in two groups. \emph{Valuation-quality baselines:} (1)
KNN-Shapley~\cite{10.14778/3342263.3342637}, a closed-form
KNN-based Shapley estimator with $\mathcal{O}(\Ntrain \log \Ntrain)$
per-validation cost; (2) \gshapley~\cite{ghorbani2019data}, the
gradient-based Shapley estimator that ZK-DV adopts as its valuation
primitive; (3) Data Shapley~\cite{ghorbani2019data}, a truncated
Monte-Carlo Shapley estimator over a fixed predictor; and (4) Data
Banzhaf~\cite{wang2023data}, the Banzhaf-value variant. \emph{ZK baselines:} (5)
\knnzk{}, a direct SNARK compilation of KNN-Shapley through the same
sumcheck/$\PCS$ pipeline as \tool; and (6) \zklsnaive{}, the per-pair
tensor encoding from \S\ref{sec:method:bucket}. We also
attempted to include SP1~\cite{sp1}, a state-of-the-art
general-purpose zkVM, but it failed to complete a single proof on the
smallest dataset in the benchmark; we report the run alongside the
ZK-DV comparison in \S\ref{sec:eval:end_to_end}.

\parh{Tasks and Metrics.}~Following prior work~\cite{wang2023privacy,
ghorbani2019data}, we evaluate two standard data-valuation tasks:
\emph{mislabeled} data detection and \emph{noisy} data detection. Of
each dataset, we hold out $10\%$ as the validation cohort and train
on the remaining $90\%$, of which a further $10\%$ is corrupted
(label flip for mislabel, feature perturbation for noisy); a faithful
estimator should rank these corrupted points lowest. We report two
data-valuation metrics: AUROC for valuation quality, and valuation
runtime for throughput. For \tool and the two ZK baselines,
we additionally
report proving time, verification time, and proof size measured on
the same hardware; \texttt{OOM} marks runs that exceeded the host
memory budget. Every reported number is the median of three
frozen-seed runs.

\parh{Implementation.}~\tool is implemented in Rust (36K LOC) over the
Goldilocks base field ($p = \goldilocks$) with a Brakedown polynomial
commitment scheme~\cite{golovnev2023brakedown}. All experiments run
on a single 64-core Intel-class server with 256\,GB RAM and Ubuntu
Linux.

\subsection{RQ1: End-to-End Capability on Real-World Data}
\label{sec:eval:end_to_end}

\begin{table*}[!t]
    \caption{Valuation quality on the 12-dataset benchmark. AUROC
    reported for both mislabel and noisy detection.}
    \label{tab:plaintext_quality}
    \centering
    \resizebox{0.8\textwidth}{!}{
    \begin{tabular}{l rr ccccc ccccc}
        \toprule
        & & & \multicolumn{5}{c}{Mislabel AUROC} & \multicolumn{5}{c}{Noisy AUROC} \\
        \cmidrule(lr){4-8} \cmidrule(lr){9-13}
        Dataset & $\Ntrain$ & $d$
        & \gshapley & DS & DB & KNN & LSH-Shapley
        & \gshapley & DS & DB & KNN & LSH-Shapley \\
        \midrule
        \TableOneAUROCRows
        \bottomrule
    \end{tabular}
    }
\end{table*}

\noindent We assess \tool on two facets in turn: first, valuation
quality on the 12-dataset benchmark; second, end-to-end proving cost
against the two ZK baselines on the same datasets.

\parh{Valuation Quality and Throughput.}
Tables~\ref{tab:plaintext_quality} and~\ref{tab:plaintext_runtime}
jointly report per-method AUROC and valuation runtime on the
12-dataset benchmark. \tool consistently outperforms baselines such as G-Shapley, Data Shapley, and Data Banzhaf across most settings, particularly in noisy detection at high feature dimensions. It remains comparable to KNN-Shapley within
$|\Delta\mathrm{AUROC}| \leq \rqOneAbsDeltaAurocMax$ on every
dataset and on both tasks, confirming that the bucket-counting
reformulation preserves the underlying valuation primitive. On throughput, \tool runs in $0.11$--$5.93$\,s
per validation cohort across all twelve datasets, whereas
KNN-Shapley itself takes $4.4$--$427$\,s. The gap widens at high
feature dimension because KNN-Shapley scales with
$\Ntrain \cdot \Ntest \cdot d$, whereas \tool's per-class bucket histograms
decouple the complexity to $(\Ntrain + \Ntest) \cdot d$. Data Shapley and
Data Banzhaf time out on the larger rows of
Table~\ref{tab:plaintext_runtime}: their cost is dominated by an
$\mathcal{O}(M \cdot \Ntrain)$ retraining loop ($M$ Monte-Carlo
permutations or coalition samples, each requiring a fresh model fit),
consistent with prior
work~\cite{wang2023privacy,ghorbani2019data,wang2023data} that reports
these estimators only at thousand-scale $\Ntrain$. \tool is
therefore the only estimator in Table~\ref{tab:plaintext_quality} that
simultaneously reaches the KNN-Shapley quality frontier and the
\gshapley-class throughput band.

\parh{End-to-End Proving Cost.}
Table~\ref{tab:zkp_cost} reports proving time $P$, verification time
$V$, and proof size $|\pi|$ for \tool, \knnzk{}, and \zklsnaive{}
across the 12-dataset benchmark. Against \knnzk{}, \tool delivers a uniform
proving-time speedup on every reached dataset, peaking at
$\rqOneBestSpeedup\times$ on \texttt{\rqOneBestSpeedupDs}. Against
\zklsnaive{}, \tool delivers
$\rqOneOrigMinSpeedup\times$--$\rqOneOrigPeakSpeedup\times$
proving-time speedups, with the verifier under
$\rqOneVerifySMax$\,s on every dataset and the proof size shrinking
in lockstep with the prover-side reduction. Where the workload
crosses out of the small regime, the comparison stops being a speed
question and becomes a feasibility question: \knnzk{} runs out of
host memory on $\rqOneKnnOomCount$ of $\rqOneKnnOomTotal$ datasets,
and \zklsnaive{} runs out on $\rqOneOrigOomCount$ datasets
(\rqOneOrigOomList), all of which are the high-$d$ OpenML and
ViT-embedding rows that a production marketplace would actually
ship. \tool is the only system in Table~\ref{tab:zkp_cost} that
verifies every one of these large-scale datasets, in
$\rqOneOpenMLProveSLo$--$\rqOneOpenMLProveSHi$\,s on the ten OpenML
rows and $\rqOneViTProveSLo$--$\rqOneViTProveSHi$\,s on the two ViT
rows, with verification under $\rqOneVerifySMax$\,s holding flat
throughout. The end-to-end picture is therefore that \tool is
strictly better than every baseline wherever a head-to-head
comparison is possible, and is the only system that reaches the
ML-embedding-scale workloads at all.

\parh{Comparison against ZK-DV and SP1.}~Direct benchmarking is currently unfeasible for two alternative ZK frameworks: ZK-DV, whose source
code is not yet open-sourced, and SP1, a general-purpose zkVM. ZK-DV~\cite{liu2026bridging} reports a maximum configuration of
$\Ntrain = 256$ on MNIST under a LeNet-5 G-Shapley substrate, on
which their optimal prover takes $\approx 2992$\,s per seller. In contrast, on \texttt{mnist}  with a ViT at
$\Ntrain = 60{,}000$ ($234\times$ larger),
\tool generates a proof in $\rqOneViTProveSLo$\,s, achieving a two-to-three order-of-magnitude speedup despite the larger workload. For SP1, we compiled LSH-Shapley as standard
Rust inside within its execution environment; on the smallest dataset
($\Ntrain = 5404$, $d = 5$) it consumed more than $87$\,GB of
resident memory and failed to complete a single proof in over one
hour. This is consistent with the structural mismatch between
general-purpose zkVMs, whose cycle counts scale with program
execution, and the structured workload LSH-Shapley targets.

\textbf{Takeaway for RQ1.} \tool is the only evaluated system to successfully complete the full benchmark suite, maintaining high fidelity with an AUROC within $\rqOneAbsDeltaAurocMax$ of the baseline. By achieving seconds-range proving times for both OpenML and ViT embeddings, \tool demonstrates the practical feasibility of providing scalable, high-accuracy guarantees for real-world data marketplaces.

\begin{table}[t]
    \caption{Valuation runtime on the 12-dataset benchmark.}
    \label{tab:plaintext_runtime}
    \centering
    \small
    \begin{tabular}{l rrrrr}
        \toprule
        & \multicolumn{5}{c}{Valuation Runtime (s)} \\
        \cmidrule(lr){2-6}
        Dataset & \gshapley & DS & DB & KNN & LSH-Shapley \\
        \midrule
        \TableOneTimeRows
        \bottomrule
    \end{tabular}
\end{table}

\begin{table*}[t]
    \caption{End-to-end ZK cost on the 12-dataset benchmark. $P$,
    $V$, $|\pi|$ are proving time, verification time, and proof size.}
    \label{tab:zkp_cost}
    \centering
    \small
    \begin{tabular}{lrr rrr rrr rrr}
        \toprule
        & & & \multicolumn{3}{c}{\knnzk} & \multicolumn{3}{c}{\zklsnaive{}} & \multicolumn{3}{c}{\tool} \\
        \cmidrule(lr){4-6} \cmidrule(lr){7-9} \cmidrule(lr){10-12}
        Dataset & $\Ntrain$ & $d$ & $P$ (s) & $V$ (s) & $|\pi|$ (MB) & $P$ (s) & $V$ (s) & $|\pi|$ (MB) & $P$ (s) & $V$ (s) & $|\pi|$ (MB) \\
        \midrule
        \TableOneZKPRows
        \bottomrule
    \end{tabular}
\end{table*}

\subsection{RQ2: Scalability}
\label{sec:eval:scalability}

\noindent We evaluate \tool along the four workload axes the
\S\ref{sec:method:bucket:complexity} cost statement covers: the
training-set size $\Ntrain$, the feature dimension $d$, the number
of classes $\Nclass$, and the validation-set size $\Ntest$. We
anchor the experiments at a reference configuration of
$\Ntrain = 16{,}384$, $d = 8$, $\Nclass = 10$, $\Ntest = 64$, and
vary one axis at a time around it. Figure~\ref{fig:scalability}
plots proving time, verification time, and proof size against each
axis.

\parh{Effect of $\Ntrain$.}
The $\Ntrain$ panel of Figure~\ref{fig:scalability} sweeps
$\Ntrain \in \{2^{4}, 2^{6}, \dots, 2^{16}\}$ at the reference
$d = 8, \Nclass = 10$. Prove time scales linearly with the dataset size for larger $\Ntrain$, reaching
$\rqTwoNTopProveS$,s at $\Ntrain = \rqTwoNTopN$. At smaller values of $\Ntrain$, the
fixed cost of the per-oracle $\PCS$ commit/open dominates, but this amortizes once
the witness population crosses the per-oracle threshold. The overall linear scaling
matches the histogram-consistency cost
$\mathcal{O}(\Ntables \cdot \Ntrain \cdot \mathrm{poly}(\Kdepth))$
of \S\ref{sec:method:bucket:complexity}.

\parh{Effect of $d$.}
\label{sec:eval:scalability:d}
The $d$ panel of \F~\ref{fig:scalability} sweeps
$d \in \{4, 16, 64, 256, 1024, 4096\}$ at the reference
$\Ntrain = 16{,}384, \Nclass = 10$. Prove time stays in the
$\rqTwoDLowProveS$\,s to $\rqTwoDHighProveS$\,s band across four
orders of magnitude in $d$, with a discrete step at $d = 64$ where
the optional PCA projection activates and adds a fixed verifiable
dimensionality-reduction stage. At $d = 4096$, \tool generates a proof in $\rqTwoDHighProveS$\,s. The pattern matches
\S\ref{sec:method:bucket:complexity}: the Shapley layers are
$d$-independent, and $d$ enters only through the hashing input and
the optional PCA projection.

\parh{Effect of $\Ntest$.}
The $\Ntest$ panel of Figure~\ref{fig:scalability} sweeps
$\Ntest \in \{32, 128, 512, 2048, 4096\}$ at the reference
$\Ntrain = 16{,}384, d = 8, \Nclass = 10$. Prove time is flat
within $\rqTwoNtestSpreadPct\%$ across the range. This follows
\S\ref{sec:method:bucket:complexity}: the histogram and weight
stages are $\Ntest$-independent, so an operator can extend the
validation set at no prover-side cost. A residual $\log \Ntest$
slope appears only in the verify-time row.

\parh{Effect of $\Nclass$.}
The $\Nclass$ panel of Figure~\ref{fig:scalability} sweeps
$\Nclass \in \{2, 5, 10, 50, 100\}$ at the reference
$\Ntrain = 16{,}384, d = 8$. At $\Nclass = \rqTwoCTopC$, the proof is
$\rqTwoCTopProofMB$\,MB and proving time is $\rqTwoCTopProveS$\,s.
The observed scaling matches the two-component cost model of
\S\ref{sec:method:bucket:complexity}: a linear cost from the interactive sumcheck messages (handling per-class lookups and weight identities), and a PCS-opening cost reduced to $O(\sqrt{\Nclass})$ via super-oracle batching. At small $\Nclass$ the proof is
floored by $\Nclass$-independent costs (the hashing module and the
fixed-overhead PCS commitments), so growth is shallow; as $\Nclass$
rises, the linear transcript term takes over and scaling approaches
linear from below. Prove time follows the same regime transition.

\begin{figure*}[t]
    \centering

    \begin{subcaptionblock}{0.49\linewidth}
        \centering
        \includegraphics[width=\linewidth]{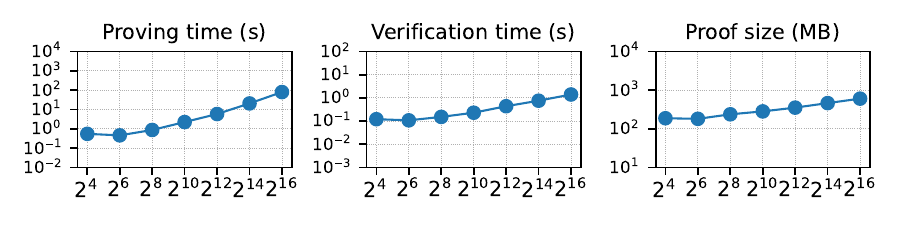}
        \vspace{-2em}
        \caption{Effect of $\Ntrain$ at $d = 8$, $\Nclass = 10$. }
        \label{fig:scalability:n}
    \end{subcaptionblock}
    \hfill
    \begin{subcaptionblock}{0.49\linewidth}
        \centering
        \includegraphics[width=\linewidth]{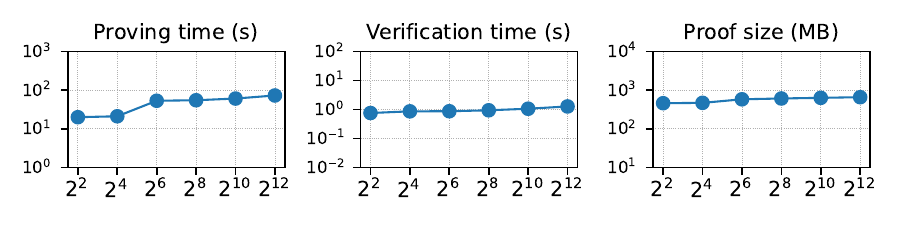}
        \vspace{-2em}
        \caption{Effect of $d$ at $\Ntrain = 16{,}384$, $\Nclass = 10$.}
        \label{fig:scalability:d}
    \end{subcaptionblock}
    \\[0.3ex]
    \begin{subcaptionblock}{0.49\linewidth}
        \centering
        \includegraphics[width=\linewidth]{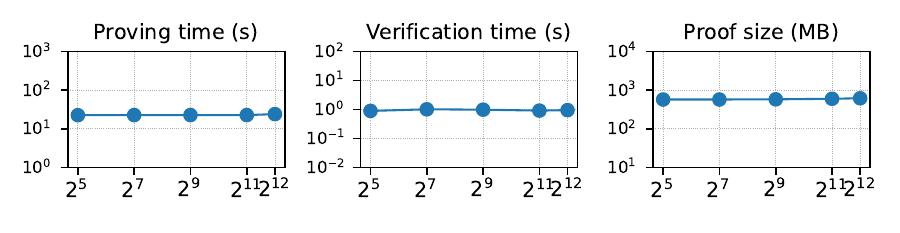}
        \vspace{-2em}
        \caption{Effect of $\Ntest$ at $\Ntrain = 16{,}384$, $d = 8$, $\Nclass = 10$.}
        \label{fig:scalability:t}
    \end{subcaptionblock}
    \hfill
    \begin{subcaptionblock}{0.49\linewidth}
        \centering
        \includegraphics[width=\linewidth]{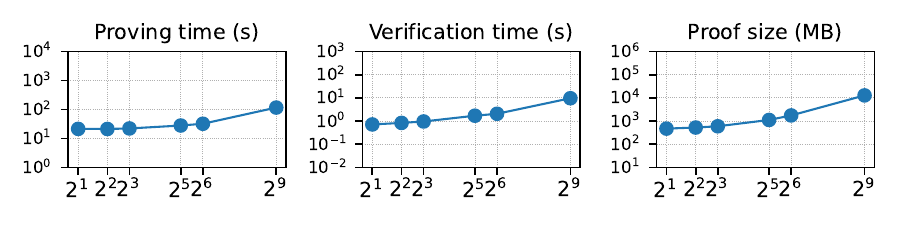}
        \vspace{-2em}
        \caption{Effect of $\Nclass$ at $\Ntrain = 16{,}384$, $d = 8$. }
        \label{fig:scalability:c}
    \end{subcaptionblock}
    \caption{\tool scalability along the four workload axes
    $(d, \Ntrain, \Ntest, \Nclass)$. }
    \label{fig:scalability}
\end{figure*}

\textbf{Takeaway for RQ2.} \tool demonstrates robust scalability, maintaining prover latency within a consistent seconds-to-tens-of-seconds range across modern ML configurations. This stability ensures that \tool's performance remains predictable and practical as system complexity grows.

\subsection{RQ3: Optimization Effectiveness and Module Breakdown}
\label{sec:eval:ablation}

\noindent To isolate the contribution of each optimization \tool
stacks on top of \zklsnaive{}, we run a one-at-a-time knockout
on a synthetic Gaussian workload at $\Ntrain = 10{,}000$, $d =
16$, $\Nclass = 10$.

\parh{Bucket-counting reformulation.}
Table~\ref{tab:ablation} reports the largest single shift on the
proving-time column: replacing the bucket-histogram refactor with the
\zklsnaive{} tensor encoding inflates proving time by
$\rqThreeBucketSpeedup\times$ ($\rqThreeFullProveS \to
\rqThreeOriginalProveS$\,s) at the same workload. The shift is consistent with the cost analysis of
\S\ref{sec:method:bucket:complexity}: the tensor encoding commits the
$\Ntables \cdot \Ntest \cdot \Ntrain$ membership tensor in full,
whereas the bucket refactor commits only the $\Nclass \cdot
\Nbuckets$ per-class histograms that the LSH-Shapley computation
actually consumes, which is what makes the bucket-counting
reformulation the primary structural lever for proving time.

\parh{Super-oracle batching.}
On the proof-size column, Table~\ref{tab:ablation} shows that
removing super-oracle PCS batching inflates the proof from
$\rqThreeFullProofMB$\,MB to $\rqThreeNoSuperOracleProofMB$\,MB
($\rqThreeSuperOracleProofRatio\times$); the corresponding
proving-time inflation is $\rqThreeFullProveS \to
\rqThreeNoSuperOracleProveS$\,s. The pattern reproduces the
$\sqrt{\Nclass}$ scaling rule of \S\ref{sec:method:bucket:complexity}:
opening the per-class oracles individually pushes $\PCS$-open work
to scale linearly with $\Nclass$ instead of with its square root. 

\parh{Sparsity skipping.}
Sparsity skipping cuts end-to-end proving time from
$\rqThreeNoSparsityProveS$\,s to $\rqThreeFullProveS$\,s
($\rqThreeSparsityReductionPct\%$) Mechanically, a support mask
exposes only the nonzero cells of the three Layer-1 histogram
sumchecks (training-histogram, per-class training-histogram, and
per-class validation-histogram), eliminating the dense-grid sweep
over empty buckets.

\parh{Per-module breakdown.}
\F~\ref{fig:module_share} partitions \tool's proving time across the five components detailed in \S\ref{sec:method:bucket}. At the \S\ref{sec:eval:ablation} reference configuration, the hashing module dominates the total cost, driven by the high ZKP overhead of sign-bit and range-proof operations. The four downstream layers collectively account for the remaining $\rqThreeBreakdownLayersPct\%$.

\begin{table}[t]
    \caption{Ablation study}
    \label{tab:ablation}
    \centering
    \small
    \begin{tabular}{lrrr}
        \toprule
        Variant & Proof (MB) & Prove (s) & Verify (s) \\
        \midrule
        Full system (\tool)            & \rqThreeFullProofMB          & \rqThreeFullProveS          & \rqThreeFullVerifyS          \\
        \,w/o super-oracle batching    & \rqThreeNoSuperOracleProofMB & \rqThreeNoSuperOracleProveS & \rqThreeNoSuperOracleVerifyS \\
        \,w/o sparsity skipping        & \rqThreeNoSparsityProofMB    & \rqThreeNoSparsityProveS    & \rqThreeNoSparsityVerifyS    \\
        \,w/o bucket reformulation     & \rqThreeOriginalProofMB      & \rqThreeOriginalProveS      & \rqThreeOriginalVerifyS      \\
        \bottomrule
    \end{tabular}
\end{table}

\begin{figure}[t]
    \centering
    \includegraphics[width=1\linewidth]{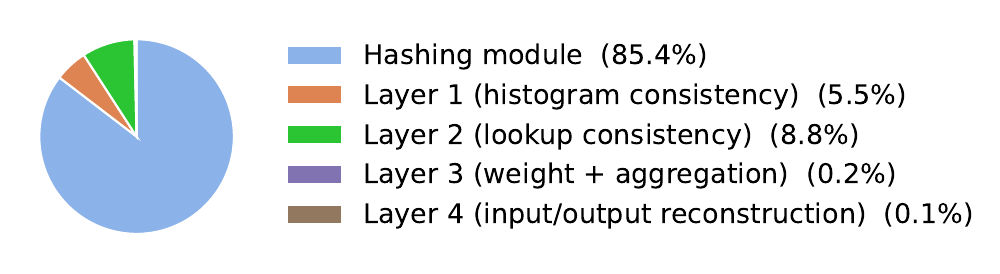}
    \vspace{-2.5em}
    \caption{\tool proving-time decomposition.}
    \label{fig:module_share}
\end{figure}

\textbf{Takeaway for RQ3.} The ablation study confirms that \tool's optimizations are orthogonal and collectively essential for achieving practical performance. The absence of these techniques leads to multi-fold degradation across these independent metrics, demonstrating that \tool's efficiency stems from a comprehensive architectural approach to verifiable valuation.

\section{Related Work}
\label{sec:related-work}

\textbf{Privacy-preserving and verifiable data valuation.}
The Shapley formulation of data valuation has produced a rich
algorithmic
line~\cite{ghorbani2019data,10.14778/3342263.3342637,kwon2022beta,wang2023data,yoon2020data,ghorbani2020distributional,jiang2023opendataval},
all of which assume a trusted valuator with plaintext access to seller
data. A second line addresses the resulting privacy tension via
differentially private estimators~\cite{wang2023privacy} or secure-MPC
prototypes~\cite{tian2022private,peng2025reliable}. While these approaches ensure privacy, they lack verifiability, i.e., the ability for external parties to verify results without accessing raw data.

\textbf{Data marketplaces and pricing.}
The database community has built a layered stack for data and model
markets, from query-based
pricing~\cite{koutris2015query,deep2017qirana} and model-based
pricing for ML~\cite{chen2019towards,liu2021dealer} to
market-platform abstractions and unifying
surveys~\cite{fernandez2020data,pei2023data}, and, more recently,
Shapley-driven pricing exposed to users~\cite{zhu2024dataprice},
auction- and competition-based pricing under federated
learning~\cite{li2024performance,sun2024profit}, and explicit
treatments of marketplace privacy and
security~\cite{alabi2025privacy}. All assume a trusted broker who sees
raw seller data in the clear. \tool removes the trust assumption  by
promoting the operator to a prover whose valuation, pricing, and
selective disclosure are each attested by a succinct proof.

\textbf{Zero-knowledge proof systems for ML and analytics.}
ZK proofs are now practical for non-trivial computations, spanning
CNN/Transformer training and inference~\cite{liu2021zkcnn,weng2021mystique,sun2024zkdl,chen2024zkml,garg2023experimenting,hao2024scalable},
billion-parameter LLMs~\cite{sun2024zkllm,qu2025zkgpt}, model
fairness~\cite{zhang2025fairzk}, ad-hoc
SQL~\cite{li2023zksql,gu2025poneglyphdb}, and causal
analytics~\cite{wang2025privacy}, with general-purpose
zkVMs~\cite{sp1,risc0} as substrate.
Rather than encoding a generic ML or SQL workload, \tool co-designs
a valuation primitive (LSH-Shapley) with a specialized ZK protocol
whose super-oracle batching and sparsity skipping make Shapley-style
data valuation provable end-to-end at marketplace scale.

\section{Discussion and Future Work}
\label{sec:future_work}

\parh{Privacy-enhancing composition.}~\tool is designed to compose orthogonally with the privacy-preserving mechanisms discussed in \S\ref{sec:related-work}. For instance, the noise calibration techniques from DP-TKNN-Shapley~\cite{wang2023privacy} can be combined with \tool to provide a verifiable defense against membership inference attacks. This allows auditors to verify that release-time noise was applied honestly without revealing the underlying scores.

\parh{Beyond LSH-Shapley.}~LSH-Shapley is a Shapley-family proxy that
inherits the faithfulness assumption of the KNN-Shapley line. Tighter
proxies --- gradient-Shapley, influence-function approximations, or
coalition-sketch estimators with sharper concentration --- remain an open
structural target, provided their final-result correctness can be verified
without replaying data-dependent search; the same requirements extend to
regression and generative utilities beyond the classification scope of
\S\ref{sec:system}.

\section{Conclusion}
\label{sec:conclusion}

We presented \tool, a practical zero-knowledge system that resolves the fundamental conflict between privacy and verifiability in data marketplaces. By co-designing the LSH-Shapley valuation primitive with a specialized, highly optimized ZKP protocol, we bridged the scalability gap that previously made verifiable valuation infeasible at modern machine learning scales. Our evaluation confirms that \tool maintains baseline valuation fidelity while reducing proving times to seconds or minutes, producing megabyte-sized proofs verifiable in seconds. Ultimately, \tool provides a deployable, auditable foundation for transparent cross-organizational data trading.

\bibliographystyle{ACM-Reference-Format}
\bibliography{refs}

\clearpage
\appendix

\section{Datasets and Pre-processing}
\label{app:datasets}

\parh{Real-world panel.}~The 12-dataset panel comprises ten OpenML
tabular benchmarks (\texttt{phoneme}, \texttt{2dplanes},
\texttt{click}, \texttt{wind}, \texttt{cpu}, \texttt{creditcard},
\texttt{fraud}, \texttt{pol}, \texttt{vehicle}, \texttt{apsfail}; raw
feature dimensions $\draw \in [5, 170]$) and two ViT embedding
datasets on \texttt{mnist} and \texttt{cifar10} ($\draw = 768$). Every
row passes through the same pre-processing pipeline: (i) class-
balanced sub-sampling down to the largest power-of-two cardinality
not exceeding the configured budget for $\Ntrain$ and $\Ntest$, (ii)
$z$-score normalization, (iii) $\ell_2$ normalization. The
power-of-two constraint is required so that the multilinear extensions used inside
the proof system live on Boolean hypercubes of matched size.
Datasets whose raw dimension exceeds the hashing budget enter the
verifiable PCA preamble described in
\S\ref{sec:method:bucket}.

\parh{Synthetic Gaussian workloads.}~The scalability sweeps in
\S\ref{sec:eval:scalability} and the ablation in
\S\ref{sec:eval:ablation} use a synthetic Gaussian generator that
gives full control over $\Ntrain$, $\draw$, $\Nclass$, and $\Ntest$.
Concretely: features $x_i \in \mathbb{R}^{\draw}$ are drawn i.i.d.\
from $\mathcal{N}(0, I_{\draw})$; we sample $\Nclass$ random unit
anchors $a_c \in \mathbb{R}^{\draw}$ and assign each point the label
$\arg\max_c \langle x_i / \|x_i\|, a_c \rangle$ (nearest-cosine
assignment); the resulting features and labels are split into
$(x_{\mathrm{train}}, y_{\mathrm{train}})$ and
$(x_{\mathrm{test}}, y_{\mathrm{test}})$ at the configured sizes,
each row $\ell_2$-normalized to match the LSH preprocessing
pipeline. The non-power-of-two cases used in the ablation reference
are zero-padded up to the next power of two, which is the natural
input format for the per-class label-indicator sparsity that the M2
fallback path of \S\ref{sec:opt:sparsity} exploits.

\section{Algorithm Pseudocode and Extended Derivations}
\label{app:extended}

\subsection{\zkls{} Prove/Verify Pseudocode}
\label{app:pseudocode}

\paragraph{Committed oracles.} Table~\ref{tab:oracle-inventory} catalogs every multilinear oracle the prover commits in Phase~1 of Algorithm~\ref{alg:lshbucket}: $13 + 2\Nclass$ bucket-specific oracles ($17$ at $\Nclass = 2$), the two hashing-preamble oracles ($\signbit$ and $\mathsf{abs\_dp}$), and, conditionally, two PCA oracles when $\draw$ exceeds the hashing budget. Each is committed as a multilinear extension in the sense of \S\ref{sec:bg:zkp}, so the verifier interacts with it only through PCS openings at Fiat--Shamir-sampled points.

\begin{table}[t]
\small
\caption{Committed oracles of \zkls{} ($\Nclass = 2$). ``Family'' column: H = histogram, L = lookup, W = weight, A = advice, Y = label indicator, P = preamble, PCA = conditional.}
\label{tab:oracle-inventory}
\begin{tabular}{@{}lll@{}}
\toprule
\textbf{Oracle} & \textbf{Domain} & \textbf{Family} \\
\midrule
$\cnt[\ell, b]$ & $\Ntables \cdot \Nbuckets$ & H \\
$\cnttr[\ell, b, c]$ & $\Ntables \cdot \Nbuckets \cdot \Nclass$ & H \\
$\cntte[\ell, b, c]$ & $\Ntables \cdot \Nbuckets \cdot \Nclass$ & H \\
$\mhat[\ell, i]$ & $\Ntables \cdot \Ntrain$ & L \\
$\mchat[\ell, i]$ & $\Ntables \cdot \Ntrain$ & L \\
$\tchat[\ell, i]$ & $\Ntables \cdot \Ntrain$ & L \\
$\harm{M}[\cdot]$ & $\Ntables \cdot \Ntrain$ & L (self-verified) \\
$W[\ell, i]$, $D, D^{-1}, e$ & $\Ntables \cdot \Ntrain$ & W \\
$\mathsf{wt\_num}, t_{\mathrm{match}}$ & $\Ntables \cdot \Ntrain$ & A \\
$y_{\mathrm{ind},c}$ & $\Ntrain$ & Y ($\Nclass$ oracles) \\
$y_{\mathrm{test\_ind},c}$ & $\Ntest$ & Y ($\Nclass$ oracles) \\
$\signbit[\ell, k, \cdot]$ & $\Ntables \cdot \Kdepth \cdot (\Ntrain{+}\Ntest)$ & P \\
$\mathsf{abs\_dp}[\ell, k, \cdot]$ & $\Ntables \cdot \Kdepth \cdot (\Ntrain{+}\Ntest)$ & P \\
$\xraw$ (conditional) & $\Ntrain \cdot \draw$ & PCA \\
$\Vpca$ (conditional) & $\draw \cdot \dpca$ & PCA \\
\bottomrule
\end{tabular}
\end{table}

The non-obvious commitment choices follow. The Shapley-weight family $W, D = M(M-1), D^{-1}, e$ records the per-table contribution, denominator, denominator inverse, and edge-case flag ($e = 1 \Leftrightarrow \mhat \leq 1$); committing the advice oracles $\mathsf{wt\_num}, t_{\mathrm{match}}$ (derivable from the rest) lets the weight-identity sumcheck stay low-degree. The label indicators $y_{\mathrm{ind},c}, y_{\mathrm{test\_ind},c}$ one-hot-encode class membership into $2\Nclass$ oracles, the factor that drives the super-oracle batching of \S\ref{sec:optimizations}. The hashing preamble commits $\signbit$ and $\mathsf{abs\_dp} = |\langle r_\ell, x\rangle|$ so the verifier can reconstruct each signed dot product at the final challenge, and the conditional PCA oracles $\xraw, \Vpca$ are committed only when $\draw$ exceeds the hashing budget.

\begin{algorithm}[t]
\small
\caption{\zkls{} $\mathsf{Prove}$ / $\Verify$ (sketch)}
\label{alg:lshbucket}
\begin{algorithmic}[1]
\State \textbf{Input (Prover):} $\{x_{i}^{\mathrm{tr}}, y_i^{\mathrm{tr}}\}_{i \in [\Ntrain]}$, $\{x_{t}^{\mathrm{te}}, y_t^{\mathrm{te}}\}_{t \in [\Ntest]}$, hash tables $\{\hashkey_\ell\}_{\ell \in [\Ntables]}$.
\State \textbf{Input (Verifier):} public parameters $\pp$, commitments $\{\commitment{i}\}_{i \in [m]}$, $\commitment{\mathrm{val}}$.
\Statex
\State \textbf{Phase 1 --- Commit.}
\State \hskip1em Prover hashes data and builds the bucket-,
  lookup-, weight-, advice-, label-, and preamble-family oracles
  of Table~\ref{tab:oracle-inventory} (with conditional PCA
  oracles when $\draw > \mathrm{budget}$).
\State \hskip1em Prover publishes $\cm{\cdot}$ for each oracle above and broadcasts $\svsum, \svavg$.
\Statex
\State \textbf{Phase 2 --- Sumcheck.}
\State \hskip1em Fiat--Shamir-derive $(\beta[0..9], \gamma[0..2], r_c, \rho, \rho_1)$ from the transcript.
\State \hskip1em Run $\mathsf{RLC}$-batched sumcheck over
  $(\Ntables, \Ntrain)$ folding the ten Layer-1/2/3 modules of
  Table~\ref{tab:module-spec} (LN batch). The per-class
  validation-histogram sumcheck is folded in via $r_c$, with the
  same $r_c$ reused for both the per-class training- and
  validation-histogram class-folding and the batched run.
\State \hskip1em Run $\mathsf{RLC}$-batched sumcheck over $(\Ntables, \Nbuckets)$: bucket-size lookup auxiliary.
\State \hskip1em If high-$\draw$: run training-side and test-side PCA-projection sumchecks plus the orthonormality sumcheck. Always-on: dot-product sign sumchecks (training and test) over $\log d$.
\State \hskip1em Run Op1 input-side reconstruction sumcheck (\S\ref{sec:method:bucket:sumchecks}, Layer~4).
\State \hskip1em Run Op7 output-side reconstruction sumcheck (\S\ref{sec:method:bucket:sumchecks}, Layer~4).
\Statex
\State \textbf{Phase 3 --- Opening.}
\State \hskip1em At each final challenge point the verifier
  requests, in the pre-optimization baseline:
  (a)~$2\Kdepth + 2$ sign-bit openings ($22$ at $\Kdepth = 10$;
  \S\ref{sec:optimizations} reduces this to $O(1)$);
  (b)~$(2\Nclass + 2)$ per-class histogram openings;
  (c)~$\Nclass$ per-class lookup openings. Each opening invokes
  one $\PCS$ opening proof.
\State \hskip1em For each provider $i$: PCS openings of $\MLExtrain$ at $\tau_1$ and $\cm{\subD{i}}$ at $\tau_{1,\mathrm{low}}$ (Op1); $\svavg$ at $\tau$ and $\cm{\subphi{i}}$ at $\ptlow$ (Op7).
\State \textbf{$\Verify$} checks all per-round polynomial identities, PCS opening proofs, Op1/Op7 final-point equalities, and accepts iff all checks pass.
\end{algorithmic}
\end{algorithm}

\paragraph{Full per-module protocol.} Algorithm~\ref{alg:lsh-bucket-full} below is the detailed counterpart of the high-level Algorithm~\ref{alg:lsh-bucket} in \S\ref{sec:method:bucket:sumchecks}. Each step in the main-body algorithm expands into the per-module commitments and sumchecks listed here.

\begin{algorithm}[t]
\small
\caption{Full \zkls{} protocol with per-module sumchecks (detailed version of Algorithm~\ref{alg:lsh-bucket}).}
\label{alg:lsh-bucket-full}
\begin{algorithmic}[1]
\Statex \textbf{Public input:} public parameters $\pp$, projection vectors $\{r_\ell^{(k)}\}$, per-provider data commitments $\{\cm{\subD{i}}\}$, validation-set commitment $\cm{\subD{\mathrm{val}}}$, hash parameters $(\Ntables, \Kdepth)$.
\Statex \textbf{Witness:} per-provider data shards $\{\subD{i}\}$, validation set $\subD{\mathrm{val}}$.
\Statex \textbf{Auxiliary input:} bucket histograms $\cnt, \cnttr, \cntte$; per-sample lookups $\mhat, \mchat, \tchat$; weight oracles $W, e, D_{\mathrm{inv}}, \mathsf{wt\_num}$; label indicators $y_{\mathrm{ind}}, y_{\mathrm{test\_ind}}$; aggregates $\svsum, \svavg$; conditional PCA witnesses $\xraw, \Vpca$ and remainders.
\Statex \textbf{Output:} proof $\pi$ certifying $\svavg = \mathsf{LSH\text{-}Shapley}(\subD{}, \subD{\mathrm{val}})$.
\Statex
\State Prover commits the feature MLEs $\MLExtrain, \MLExtest$ (and $\xraw, \Vpca$, remainders if PCA enabled).
\State Prover commits hashing-module oracles $\signbit, \mathsf{abs\_dp}, \dotprod$ (and ZK-PCA remainder oracles if PCA enabled).
\State Prover and verifier run the hashing-module sumchecks (sign-boolean, sign-range lookup, dot-product binding for both training and validation), and the ZK-PCA sumchecks if PCA is enabled.
\State Prover commits the bucket-histogram oracles $\cnt, \cnttr, \cntte$, the per-sample lookups $\mhat, \mchat, \tchat$, and the per-class label indicators $y_{\mathrm{ind}}, y_{\mathrm{test\_ind}}$.
\State Prover and verifier run the histogram-consistency sumchecks.
\State Prover and verifier run the lookup-consistency sumchecks.
\State Prover commits the weight oracles $W, e, D_{\mathrm{inv}}, \mathsf{wt\_num}$ and the aggregates $\svsum, \svavg$.
\State Prover and verifier run the weight-identity sumchecks (edge case, interior case, edge-flag boolean, denominator-inverse structural), the table-aggregation sumcheck, and the normalization sumcheck.
\State Prover and verifier run the input-side reconstruction sumcheck binding $\MLExtrain$ to $\{\cm{\subD{i}}\}$, and the output-side reconstruction sumcheck binding $\cm{\svavg}$ to $\{\cm{\subphi{i}}\}$.
\State Prover sends PCS openings at all sumcheck challenge points; verifier replays the round-by-round identities and the closing equality at each challenge.
\State Each provider runs three local checks (input binding, output binding, transcript correctness) on the per-provider slice they verify.
\Statex
\State \emph{In the deployed implementation, the modules sharing a sum domain collapse into five batched sumcheck instances via random linear combination; see~\S\ref{sec:opt:super-oracle} and Appendix~\ref{app:module-table}.}
\end{algorithmic}
\end{algorithm}

\subsection{\zkls{} Per-Module Specification Table}
\label{app:module-table}

Table~\ref{tab:module-spec} catalogs the logical sumcheck modules of \zkls{} (\S\ref{sec:method:bucket:sumchecks}). For each module we list its verification layer, target-sum identity (canonical forms drawn from \S\ref{sec:method:bucket:witnesses} and the closed-form derivation of Appendix~\ref{app:closed-form}), sum domain, maximal round-message degree, and the RLC batch the module is folded into in the deployed implementation (\S\ref{sec:optimizations}). The deployed implementation collapses logical modules sharing a sum domain into five RLC-batched sumcheck instances: a batched-LN sumcheck over $(\Ntables, \Ntrain)$ folding ten same-domain identities under coefficients $\beta[0..9]$, a batched-LB sumcheck over $(\Ntables, \Nbuckets)$ folding three squared-count auxiliaries under $\gamma[0..2]$, plus three standalone sumchecks (V3, S-agg, S1).

\begin{table*}[t]
\scriptsize
\caption{Per-module specification for \zkls{}. ``RLC batch'' is the batch each module is folded into in the deployed implementation: \textsc{LN} = batched-$(\Ntables, \Ntrain)$ under $\beta[0..9]$, \textsc{LB} = batched-$(\Ntables, \Nbuckets)$ under $\gamma[0..2]$, \textsc{Std} = standalone. Identities use the canonical forms of \S\ref{sec:method:bucket:witnesses}; $\bar{w}$ denotes the denominator-cleared weight numerator (Appendix~\ref{app:closed-form}).}
\label{tab:module-spec}
\centering
\begin{tabular}{@{}p{4.2cm}lp{6.5cm}p{2.2cm}lc@{}}
\toprule
\textbf{Module (descriptive name)} & \textbf{Layer} & \textbf{Target-sum identity} & \textbf{Sum domain} & \textbf{Round-deg.} & \textbf{RLC batch} \\
\midrule
V1 (training-histogram consistency) & 1 & $\cnt[\ell, b] = \sum_h \prod_k \bigl(b_k s_k(\ell, h) + (1{-}b_k)(1{-}s_k(\ell, h))\bigr)$ & $(\Ntables, \Ntrain)$ & $\Kdepth + 1$ & \textsc{LN} \\
V2 (per-class training-histogram consistency) & 1 & $\cnttr[\ell, b, c] = \sum_h \mathbf{1}[y_h = c] \cdot \mathrm{eq}_\Kdepth(b, \hashkey_\ell(x_h))$ & $(\Ntables, \Ntrain)$ & $\Kdepth + 2$ & \textsc{LN} \\
V3 (per-class validation-histogram consistency) & 1 & $\cntte[\ell, b, c] = \sum_t \mathbf{1}[y_t^{\mathrm{te}} = c] \cdot \mathrm{eq}_\Kdepth(b, \hashkey_\ell(x_t^{\mathrm{te}}))$ & $(\Ntables, \Ntest)$ & $\Kdepth + 2$ & \textsc{Std} \\
\midrule
V4-main (bucket-size lookup, main) & 2 & $\sum_{(\ell, h)} \mathrm{eq}_\ell \cdot \prod_k g_k \cdot \mhat = \langle \mathrm{eq}, \cnt \rangle$ ($\mathrm{eq}_\Kdepth$ identity) & $(\Ntables, \Ntrain)$ & $\Kdepth + 2$ & \textsc{LN} \\
V4-aux (bucket-size lookup, squared-count auxiliary) & 2 & $C_4 = \sum_{(\ell, b)} \cnt^{\,2} \cdot \mathrm{eq}$ (squared-count) & $(\Ntables, \Nbuckets)$ & $3$ & \textsc{LB} \\
V5-main (same-class lookup, main) & 2 & $\mathrm{eq}_\Kdepth$ identity binding $\mchat \leftrightarrow \cnttr$ & $(\Ntables, \Ntrain)$ & $\Kdepth + 2$ & \textsc{LN} \\
V5-aux (same-class lookup, squared-count auxiliary) & 2 & Squared-count claim on $\cnttr \cdot \mathsf{cc\_weighted}$ & $(\Ntables, \Nbuckets)$ & $3$ & \textsc{LB} \\
V6-main (validation-class lookup, main) & 2 & $\mathrm{eq}_\Kdepth$ identity binding $\tchat \leftrightarrow \cntte$ & $(\Ntables, \Ntrain)$ & $\Kdepth + 2$ & \textsc{LN} \\
V6-aux (validation-class lookup, squared-count auxiliary) & 2 & Squared-count claim on $\cntte \cdot \mathsf{cc\_weighted}$ & $(\Ntables, \Nbuckets)$ & $3$ & \textsc{LB} \\
\midrule
S-new-2-edge (edge-case weight identity) & 3 & $\sum_{(\ell, h)} e \cdot [\,W - \mhat \cdot \mathsf{t\_match}\,] = 0$ on $e = 1$ & $(\Ntables, \Ntrain)$ & $4$ & \textsc{LN} \\
S-new-2-nonedge (interior weight identity) & 3 & $\sum_{(\ell, h)} (1 - e) \cdot [\,W \cdot D - \mathsf{wt\_num}\,] = 0$ on $M \geq 2$ & $(\Ntables, \Ntrain)$ & $4$ & \textsc{LN} \\
E-bool (edge-flag boolean check) & 3 & $e \cdot (1 - e) = 0$ (Boolean flag) & $(\Ntables, \Ntrain)$ & $2$ & \textsc{LN} \\
S-weight-1 (denominator-inverse structural check, non-edge) & 3 & $(1 - e) \cdot D \cdot D_{\mathrm{inv}} = (1 - e)$ (denominator inverse on non-edge) & $(\Ntables, \Ntrain)$ & $3$ & \textsc{LN} \\
S-weight-2 (denominator-inverse structural check, edge) & 3 & $e \cdot D = 0$ (structural; $D = 0$ on edge) & $(\Ntables, \Ntrain)$ & $2$ & \textsc{LN} \\
S-agg (table-aggregation sumcheck) & 3 & $\svsum[h] = \sum_{\ell} W[\ell, h]$ & $\Ntables$ & $1$ & \textsc{Std} \\
S1 (normalization sumcheck) & 3 & $\svavg \cdot (\Ntables \cdot \Ntest) = \svsum$ (normalization, inherited) & $\Ntrain$ & $1$ & \textsc{Std} \\
\bottomrule
\end{tabular}
\end{table*}

\subsection{LSH-Shapley Closed Form: Derivation, Axioms, and Bounds}
\label{app:closed-form}

This subsection derives the four-case closed form
\eqref{eq:phi-pm}, verifies the four Shapley fairness axioms,
and records the sign and monotonicity properties cited in
\S\ref{sec:method:plaintext:spec}.

\parh{Setup recap.} Fix a validation point
$(x_t^{\mathrm{te}}, y_t)$ and a hash table $\ell$. Write the
bucket as
$\mathcal{B}_\ell^t = \{j \in [\Ntrain] :
  B_\ell(x_j^{\mathrm{tr}}) = B_\ell(x_t^{\mathrm{te}})\}$ of
size $M$, the agreeing subset $\mathcal{B}_\ell^{t,+}$ of size
$M^{+}$, and the disagreeing subset of size
$M^{-} = M - M^{+}$. The per-table bucket-vote utility (defined
in \S\ref{sec:method:plaintext:motivation}) is
$u_\ell^t(S) = |S \cap \mathcal{B}_\ell^{t,+}| / |S \cap \mathcal{B}_\ell^t|$,
with $u_\ell^t(S) := 0$ when $S \cap \mathcal{B}_\ell^t = \emptyset$.
Throughout we write $\harm{m} = \sum_{n=1}^{m} 1/n$ ($\harm{0} = 0$).

\parh{Dummy collapse (Lemma B.1).} For any
$j \notin \mathcal{B}_\ell^t$ and any $S$,
$(S \cup \{j\}) \cap \mathcal{B}_\ell^t = S \cap \mathcal{B}_\ell^t$,
so $u_\ell^t(S \cup \{j\}) = u_\ell^t(S)$. Non-colliders are
therefore Shapley dummies and $\phi_j = 0$, and the Shapley
value of any $i \in \mathcal{B}_\ell^t$ in the $\Ntrain$-player
game equals its value in the $M$-player sub-game over
$\mathcal{B}_\ell^t$.

\parh{Marginal contributions (Lemma B.2).} For
$T \subseteq \mathcal{B}_\ell^t \setminus \{i\}$ with $|T| = k$
and $|T \cap \mathcal{B}_\ell^{t,+}| = m^{+}$, direct algebra
gives, for $k \geq 1$,
\[
\Delta_i^{+}(T) := u_\ell^t(T \cup \{i\}) - u_\ell^t(T)
  \;=\; \frac{k - m^{+}}{k(k+1)}, \qquad
\Delta_i^{-}(T) \;=\; -\frac{m^{+}}{k(k+1)},
\]
where $+$ / $-$ denote $y_i = y_t$ / $y_i \neq y_t$. For
$k = 0$: $\Delta_i^{+}(T) = 1$ and $\Delta_i^{-}(T) = 0$.

\parh{Closed form (Theorem B.1).} For $i \in \mathcal{B}_\ell^t$
and $M \geq 2$,
\[
\phipos = \frac{\harm{M}}{M}
  - \frac{(M^{+}-1)(\harm{M}-1)}{M(M-1)}, \qquad
\phineg = -\frac{M^{+}(\harm{M}-1)}{M(M-1)}.
\]
For $M = 1$, $\phi_i = \mathbf{1}[y_i = y_t]$. For
$i \notin \mathcal{B}_\ell^t$, $\phi_i = 0$.

\emph{Proof (correct-label case, $M \geq 2$, $y_i = y_t$).}
By Lemma B.1 we work in the $M$-player sub-game. Decomposing the
Shapley sum over $k = |T|$ and $m^{+} = |T \cap \mathcal{B}_\ell^{t,+}|$,
\[
\phipos = \frac{1}{M} +
  \sum_{k=1}^{M-1} \frac{k!\,(M{-}k{-}1)!}{M!}
  \sum_{m^{+}}
  \binom{M^{+}{-}1}{m^{+}}\binom{M^{-}}{k{-}m^{+}}
  \frac{k - m^{+}}{k(k+1)},
\]
where the leading $1/M$ collects the $k = 0$ term and $m^{+}$
ranges over $0, \ldots, \min(k, M^{+} - 1)$. The Vandermonde
identity and its weighted variant give
\begin{align*}
\sum_{m^{+}}\!\binom{M^{+}{-}1}{m^{+}}\!\binom{M^{-}}{k{-}m^{+}}
  &= \binom{M{-}1}{k}, \\
\sum_{m^{+}}\! m^{+}\!\binom{M^{+}{-}1}{m^{+}}\!\binom{M^{-}}{k{-}m^{+}}
  &= (M^{+}{-}1)\binom{M{-}2}{k{-}1},
\end{align*}
which collapse the inner sum to
$k\binom{M-1}{k} - (M^{+}-1)\binom{M-2}{k-1}$. Substituting and
using $\sum_{k=1}^{M-1} 1/(k+1) = \harm{M} - 1$,
\[
\phipos = \frac{1}{M} + \frac{\harm{M} - 1}{M}
  - \frac{(M^{+}-1)(\harm{M} - 1)}{M(M-1)}
  = \frac{\harm{M}}{M}
  - \frac{(M^{+}-1)(\harm{M}-1)}{M(M-1)}.
\]
The incorrect-label case is analogous, using
$\sum_{m^{+}} m^{+} \binom{M^{+}}{m^{+}}\binom{M^{-}-1}{k-m^{+}}
  = M^{+}\binom{M-2}{k-1}$, and yields $\phineg$. The $M = 1$
case follows since the only coalition is empty
($k = 0$), giving $\phi_i = \Delta_i(T = \emptyset) / 1$.
\hfill$\square$

\parh{Shapley axioms (Theorem B.2).} The closed form satisfies
all four fairness axioms against $u_\ell^t$.
\begin{itemize}[nosep,leftmargin=1.3em]
\item \emph{Efficiency.}
$\sum_{i \in [\Ntrain]} \phi_i = M^{+}/M = u_\ell^t([\Ntrain])$.
For $M \geq 2$, expanding
$M^{+} \cdot \phipos + M^{-} \cdot \phineg$ and using
$M^{+} + M^{-} = M$ telescopes the harmonic terms to $M^{+}/M$;
non-colliders contribute zero by Lemma B.1.
\item \emph{Symmetry.} $\phipos$ and $\phineg$ depend on $i$
only through $(M, M^{+}, \mathbf{1}[y_i = y_t])$, so any two
bucket-mates sharing a label receive equal value.
\item \emph{Dummy player.} Established by Lemma B.1.
\item \emph{Linearity.} Inherited from the linearity of Shapley
values in the utility function. This is what licenses the
multi-table mean of \S\ref{sec:method:plaintext:motivation}: the
Shapley value of $\bar{u} = (\Ntables \Ntest)^{-1}
\sum_{\ell, t} u_\ell^t$ equals the corresponding mean of the
per-table per-validation Shapley values.
\end{itemize}

\parh{Sign and monotonicity (Theorem B.3).} For $M \geq 2$, the
algebraic rewrite
$\phipos = 1/M + (M - M^{+})(\harm{M} - 1)/(M(M-1))$ gives:
\begin{itemize}[nosep,leftmargin=1.3em]
\item $\phipos \geq 1/M > 0$ always; equality iff $M^{+} = M$.
\item $\phineg \leq 0$ always; equality iff $M^{+} = 0$.
\item $\phipos - \phineg = 1/M + (\harm{M} - 1)/(M-1) > 0$, so
the closed form discriminates correct-label from
incorrect-label colliders.
\item Both weights are non-increasing in $M^{+}$:
$\partial \phipos / \partial M^{+}
  = \partial \phineg / \partial M^{+}
  = -(\harm{M} - 1)/(M(M-1)) \leq 0$.
\item Asymptotically, fixing $p = M^{+}/M$, both weights are
$O(\harm{M}/M) = O(\log M / M)$ as $M \to \infty$, capturing the
dilution of per-point value in large buckets.
\end{itemize}

\subsection{Theorem 1: Rank Concentration (Proof Sketch)}
\label{app:thm1-proof}

\parh{Theorem 1 (rank-concentration).} Fix a validation point $x^{\mathrm{te}}$ and a training point $x_{i}^{\mathrm{tr}}$. Let $\rho(x_{i}^{\mathrm{tr}})$ denote the true KNN-Shapley rank of $x_{i}^{\mathrm{tr}}$ against $x^{\mathrm{te}}$, and $\hat\rho(x_{i}^{\mathrm{tr}}; \Ntables)$ the LSH-Shapley rank estimated from $\Ntables$ independent tables. For any $\varepsilon, \delta > 0$ there exists $\Ntables^{\ast} = \Ntables^{\ast}(\varepsilon, \delta, \Kdepth, d)$ such that for $\Ntables \geq \Ntables^{\ast}$,
\[
\Pr\bigl[\,|\hat\rho(x_{i}^{\mathrm{tr}}; \Ntables) - \rho(x_{i}^{\mathrm{tr}})| > \varepsilon \cdot \Ntrain\,\bigr] \;\leq\; \delta.
\]

\parh{Proof sketch.} The proof is a Hoeffding-style concentration argument over the $\Ntables$ independent per-table estimators. Each per-table bucket-collision count contributes a bounded-variance, independent random variable under the hash-family randomness; averaging over $\Ntables$ tables yields sub-Gaussian tails with rate $\Ntables^{-1}$. A union bound across training points and a standard rank-stability argument (the rank function is $1$-Lipschitz in sample-wise Shapley-value perturbations up to ties) converts the per-point concentration into the stated rank-difference bound. The explicit constant tracks $\Kdepth$ and $d$ through the hash-family bias; see below for the variance-saturation refinement that replaces the Hoeffding bound with an $O(1)$-in-$\Ntrain$ estimate.

\subsection{Parameter Selection: Variance-Saturation Derivation}
\label{app:param-sel}

\parh{Hoeffding baseline for $\Ntables$.} A direct Hoeffding bound on the per-table estimator yields $\Ntables = O(\log(\Ntrain/\delta) / \varepsilon^{2})$ tables to achieve $\varepsilon$-accurate Shapley values with probability at least $1 - \delta$. This is the bound underlying Theorem~1.

\parh{Variance-saturation refinement.} The Hoeffding bound is loose because the per-table estimator variance saturates once $\Ntables$ exceeds the effective support of the bucket-collision distribution. Tracking the variance of $\phi^{\pm}$ directly --- using the bounded-variance structure of the closed-form weights --- yields a refined design rule $\Ntables^{\ast} = O(1)$ under fixed $(\varepsilon, \delta, \Kdepth)$; empirically $\Ntables = 16$ suffices across all benchmark datasets. The constant-in-$\Ntrain$ scaling is what makes LSH-Shapley ZK-feasible at marketplace scale: the prover's Layer-1 work scales linearly in $\Ntables$, so a $\log \Ntrain$ factor in $\Ntables$ would reintroduce a $\Ntrain \log \Ntrain$ term that the histogram reformulation was designed to remove.

\parh{Choice of $\Kdepth$.} $\Kdepth$ controls bucket granularity: larger $\Kdepth$ sharpens label discrimination (fewer spurious collisions between label classes), but the number of buckets $\Nbuckets = 2^{\Kdepth}$ drives Layer-2 lookup work (\S\ref{sec:method:bucket:complexity}) and drives the sign-bit opening count in Phase 3. We pick the largest $\Kdepth$ the ZK budget tolerates subject to a quality lower bound (rank-correlation against exact KNN-Shapley on a held-out validation slice); our reference $(\Ntables, \Kdepth) = (16, 10)$ configuration balances these constraints across the 12-dataset benchmark suite.

\subsection{Super-Oracle Batching: Soundness Proof}
\label{app:super-oracle-soundness}

This subsection gives the full proof of the super-oracle batching
soundness claim deferred from \S\ref{sec:opt:super-oracle}.

\parh{Setup recap.} Let $\{f_j\}_{j=0}^{Q-1}$ be a family of
multilinear sub-oracles on $\Bool^{n}$ obtained by fixing one
index axis of a parent tensor, with $Q$ padded to a power of two.
The prover commits the super-oracle
$F : \Bool^{\log Q + n} \to \EFld$ defined by $F(j, x) := f_j(x)$
and extended multilinearly. At a common upstream evaluation point
$x^{\star} \in \EFld^{n}$ already fixed by the parent sumcheck,
the verifier samples a fresh challenge
$r \in \EFld^{\log Q}$, requests the single PCS opening
$F(r, x^{\star})$, and accepts the prover-supplied values
$v_0, \ldots, v_{Q-1}$ iff~\eqref{eq:super-oracle-rlc} holds.

\parh{Claim.} For any prover $\mathcal{P}^{\star}$ and any choice
of values $\{v_j\}$, the probability over $r$ that
$\mathcal{P}^{\star}$ passes the equality check
$F(r, x^{\star}) = \sum_{j} \eq(r, j) \cdot v_j$ while having some
$v_{j^{\star}} \neq f_{j^{\star}}(x^{\star})$ is at most
$\log Q / |\EFld|$.

\parh{Proof.} Define the residual polynomial
\[
g(r) \;:=\; F(r, x^{\star}) \;-\;
  \sum_{j \in \Bool^{\log Q}} \eq(r, j) \cdot v_j.
\]
Both terms are multilinear in $r$ of total degree at most
$\log Q$. By multilinear interpolation,
$F(r, x^{\star}) = \sum_{j} \eq(r, j) \cdot f_j(x^{\star})$ for
every $r$, so
$g(r) = \sum_{j} \eq(r, j) \cdot (f_j(x^{\star}) - v_j)$.
If $v_{j^{\star}} \neq f_{j^{\star}}(x^{\star})$ for some
$j^{\star}$, then evaluating at $r = j^{\star}$ gives
$g(j^{\star}) = f_{j^{\star}}(x^{\star}) - v_{j^{\star}} \neq 0$,
so $g \not\equiv 0$. By the Schwartz--Zippel lemma over $\EFld$,
$\Pr_{r \gets \EFld^{\log Q}}[g(r) = 0] \leq \log Q / |\EFld|$.

\parh{Composition with the parent sumcheck.} The verifier plugs
the recovered $\{v_j\}$ into the parent sumcheck's closing
identity exactly as in the unbatched protocol; under the
super-oracle binding above, the parent sumcheck inherits its
pre-existing soundness error. Each of the three families
($\signbit$, $\cnttr$, $\{\mchat, \tchat\}$) is packed into a
separate super-oracle, and their soundness errors add by a union
bound. With $\log Q \leq \log(2\Nclass + 2) + 1 \le 5$ and
$|\EFld| = \goldilocks \approx 2^{64}$, the per-family
contribution is below $2^{-59}$, dominated by the parent
sumcheck and PCS-binding terms.

\parh{Endpoint correctness.} The single PCS opening
$F(r, x^{\star})$ binds the prover to the unique multilinear
extension $F$ that agrees with the committed sub-oracle slices on
$\Bool^{\log Q + n}$. Schwartz--Zippel binding therefore reaches
all the way to the underlying $\{f_j\}$ rather than only to a
prover-chosen extension, completing the chain from the PCS-binding
assumption to per-slice equality
$v_j = f_j(x^{\star})$. \hfill$\square$

\subsection{Op7 Selective-Disclosure: Full Reconstruction Sumcheck}
\label{app:op7-full}

\parh{Op1 input-side mirror.} Op1 is structurally identical to Op7 but anchors the input end of the binding chain via the training-data oracle $\MLExtrain$ over the augmented $(\log \Ntrain + \log d)$ cube, using per-provider oracle $\subD{i}$. The identity proved is
\[
\MLE{x_{\mathrm{tr}}}(h_{\mathrm{tr}}, h_d) \;=\; \sum_{i=1}^{m} \eq(h_{\mathrm{tr,high}}, \mathrm{bin}(i)) \cdot \subD{i}(h_{\mathrm{tr,low}}, h_d),
\]
reduced by a Fiat--Shamir sumcheck to one $\PCS$ opening of $\MLExtrain$ at $\tau_1$ plus $m$ openings of $\subD{i}$ at $\tau_{1,\mathrm{low}}$. Together Op1 $+$ middle sumchecks (\S\ref{sec:method:bucket:sumchecks}) $+$ Op7 form the end-to-end binding chain $D_i \to \cm{\subD{i}} \to \cm{\MLExtrain} \to \cdots \to \cm{\svavg} \to \cm{\subphi{i}} \to \svavg_{|_{\indexset{i}}}$; falsifying any link forces a $\PCS$ forgery, infeasible under PCS binding.

\parh{Three-check provider verification.} Provider $\provider_i$ executes three local checks that together close the five-attack surface (data substitution, misplacement, intermediate-computation corruption, fabricated values, cross-provider attribution).
\begin{itemize}[nosep]
  \item \textbf{Check A --- input binding.} Recompute $\cm{\subD{i}}$ locally from the raw provider data (one Brakedown encoding $+$ Merkle hashing over $\Ntrain/m \cdot d$ elements) and match the value published in the proof transcript.
  \item \textbf{Check B --- output binding.} From the $\Ntrain/m$ received Shapley values, reconstruct the MLE $\subphi{i}$ over $\log(\Ntrain/m)$ variables, recompute $\cm{\subphi{i}}$, and match the proof.
  \item \textbf{Check C --- computation correctness.} Verify the full transcript: Op1, all sumchecks of \S\ref{sec:method:bucket:sumchecks}, Op7.
\end{itemize}

\parh{Binding-chain detail.} A malicious prover who substitutes provider $i$'s data fails Check A (the committed $\cm{\subD{i}}$ does not match); a prover who tampers with intermediate computation fails Check C (a middle-sumcheck identity rejects); a prover who fabricates scores fails Check B (the reconstructed $\cm{\subphi{i}}$ does not match the Op7 output-endpoint commitment). Cross-provider attribution attacks --- swapping scores across providers --- fail Check B at the defrauded provider and fail Check A at any provider whose slice was moved. The soundness argument composes: under the PCS-binding assumption, the probability that all three checks accept on a forged witness is bounded by the union of the per-module soundness errors plus the PCS-binding error, which is $\leq 2\log \Ntrain / |\EFld|$ plus a negligible PCS term.

\parh{Per-provider opening cost.} Op1 $+$ Op7 together add $2m$ commits (one $\cm{\subD{i}}$ and one $\cm{\subphi{i}}$ per provider), two degree-$2$ sumchecks of length $\log \Ntrain$, and $2(m+1)$ $\PCS$ openings: two endpoint openings $\cm{\MLExtrain}$ at $\tau_1$ and $\cm{\svavg}$ at $\tau$ (amortized across all providers), plus $m$ provider-side openings each of $\cm{\subD{i}}$ at $\tau_{1,\mathrm{low}}$ and $\cm{\subphi{i}}$ at $\ptlow$. The $\PCS$ openings are the dominant per-provider cost; \S\ref{sec:evaluation} measures the end-to-end wall-clock overhead.

\parh{Soundness and privacy.} \emph{Soundness.} A malicious prover committing incorrect sub-polynomials is caught with probability $\geq 1 - 2\log \Ntrain / |\EFld|$ since both endpoints are bound to $\PCS$ openings. \emph{Privacy.} Each provider's $\PCS$ opening reveals a single random linear combination of the provider's $\Ntrain/m$ values (negligible for $\Ntrain/m \gg 1$). The sumcheck is not formally zero-knowledge; masking extensions achieve formal ZK at one extra commitment (Future Work, \S\ref{sec:future_work}).

\subsection{\zkls{} Soundness Chain and Per-Sumcheck Cost Table}
\label{app:soundness-chain}

This appendix expands the soundness sketch and complexity summary of \S\ref{sec:method:bucket:soundness} with a per-link reduction and a per-sumcheck cost table.

\parh{Per-link reduction.} If the training-histogram sumcheck accepts, $\cnt$ agrees with the committed per-sample bucket assignments under $B_\ell$; the per-class training-histogram and validation-histogram sumchecks anchor $\cnttr$ and $\cntte$ respectively, folding the class dimension via a single random $r_c$ challenge. The bucket-size lookup sumcheck then forces $\mhat$ to match $\cnt$ at every sample's bucket, the same-class lookup sumcheck does the same for $\mchat \leftrightarrow \cnttr$, and the validation-class lookup sumcheck for $\tchat \leftrightarrow \cntte$. The match-count advice sumcheck binds $t_{\mathrm{match}}$ to $\tchat$. With $(\mhat, \mchat, \tchat, t_{\mathrm{match}}, \harm{M})$ thus anchored, the edge-case and interior weight sumchecks together force $W$ to equal the closed-form weight identity $W = \phi^{\pm}(\mhat, \mchat, \harm{\mhat})$ on the $M \in \{0,1\}$ and $M \geq 2$ regions respectively; the table-aggregation sumcheck aggregates $W$ across the $\Ntables$ tables into $\svsum$, and normalization by $\Ntables \cdot \Ntest$ yields $\svavg$. The harmonic oracle $\harm{M}$ is checked compositionally: a forged $\harm{M}$ produces a wrong $W$ via the weight sumchecks, hence a wrong $\svavg$ via the table-aggregation sumcheck, which fails Op7's reconstruction Check~B (each provider's Check~B locally recomputes $\cm{\subphi{i}}$ from the values they receive, so a mismatch is detected without a dedicated module for $\harm{M}$).

\parh{Per-sumcheck cost table.} Table~\ref{tab:sumcheck-costs} lists each of the eleven bucket-specific sumchecks and the ancillary sumchecks together with their domain, degree, and asymptotic prover work. The total soundness error sits below $2^{-50}$ at typical parameters because $\sum_{\text{modules}} \mathrm{deg} \cdot \log(\mathrm{domain}) / |\F|$ remains negligible against the Goldilocks field size $|\F| = \goldilocks \approx 2^{64}$, with PCS binding contributing a similarly negligible additive term.

\begin{table}[t]
\centering
\scriptsize
\caption{Per-sumcheck cost table for \zkls{}. Domain shows the sumcheck variables; degree is the maximal round polynomial degree; prover work is the asymptotic per-module cost.}
\label{tab:sumcheck-costs}
\setlength{\tabcolsep}{4pt}
\begin{tabular}{@{}lccc@{}}
\toprule
Sumcheck & Domain & Degree & Prover work \\
\midrule
Training-histogram & $\log\Ntables + \log\Ntrain$ & $\Kdepth + 1$ & $O(\Ntables \cdot \Ntrain \cdot \Kdepth^{2})$ \\
Per-class training-histogram & $\log\Ntables + \log\Ntrain$ & $\Kdepth + 2$ & $O(\Ntables \cdot \Ntrain \cdot \Kdepth^{2})$ \\
Per-class validation-histogram & $\log\Ntables + \log\Ntest$ & $\Kdepth + 2$ & $O(\Ntables \cdot \Ntest \cdot \Kdepth^{2})$ \\
Bucket-size lookup & $\log\Ntables + \log\Ntrain$ & $\Kdepth + 2$ & $O(\Ntables \cdot \Ntrain)$ \\
Same-class lookup & $\log\Ntables + \log\Ntrain$ & $\Kdepth + 2$ & $O(\Ntables \cdot \Ntrain)$ \\
Validation-class lookup & $\log\Ntables + \log\Ntrain$ & $\Kdepth + 2$ & $O(\Ntables \cdot \Ntrain)$ \\
Squared-count auxiliary & $\log\Ntables + \log\Nbuckets$ & $3$ & $O(\Ntables \cdot \Nbuckets)$ \\
Match-count advice & $\log\Ntables + \log\Ntrain$ & $\Kdepth + 1$ & $O(\Ntables \cdot \Ntrain)$ \\
Edge-case weight & $\log\Ntables + \log\Ntrain$ & $3$ & $O(\Ntables \cdot \Ntrain)$ \\
Interior weight & $\log\Ntables + \log\Ntrain$ & $3$ & $O(\Ntables \cdot \Ntrain)$ \\
Table-aggregation & $\log\Ntables + \log\Ntrain$ & $2$ & $O(\Ntables \cdot \Ntrain)$ \\
\midrule
Dot-product sign (ancillary) & $\log d$ & $2$ & $O(\Ntables \cdot \Kdepth \cdot (\Ntrain + \Ntest) \cdot d)$ \\
PCA-projection (conditional) & $\log\draw + \log\dpca$ & $2$ & $O(\Ntrain \cdot \draw \cdot \dpca)$ \\
\bottomrule
\end{tabular}
\end{table}

The Layer-1 row dominates prover work; Layer-2 and Layer-3 rows are sub-dominant by a factor of $\Kdepth$ in the typical regime, and the ancillary rows confine $d$ and $\draw$ to the preamble as claimed in the body.

\section{Sparsity-Aware Sumcheck: Pseudocode and Soundness}
\label{app:sparsity-aware}

This appendix expands the sparsity-aware sumcheck path of
\S\ref{sec:opt:sparsity}: the per-round prover pseudocode, the
support-superset invariant that controls correctness, and the
soundness argument that the sparse path produces round messages
identical to the dense path.

\parh{Notation.} For a multilinear factor
$f : \Bool^{n} \to \Fld$, let
$\supp(f) = \{x \in \Bool^{n} : f(x) \neq 0\}$ be its true support
and let $\widehat{\supp}(f) \supseteq \supp(f)$ denote the prover-tracked
support mask --- a superset that the prover maintains explicitly.
After fixing the first $r$ variables to challenges
$(\rho_1, \ldots, \rho_r) \in \EFld^{r}$, write
$f^{(r)} : \Bool^{n - r} \to \EFld$ for the partially-folded
factor; the corresponding mask
$\widehat{\supp}(f^{(r)}) \subseteq \Bool^{n - r}$ is obtained by
projecting away the $r$ folded coordinates (a fold step preserves
the superset property because folding cannot create new nonzero
cells beyond the projection of the previous mask).

\parh{Support-superset invariant.} The pseudocode of
Algorithm~\ref{alg:sparse-sumcheck} maintains the following
invariant after every round $r$:
\[
\supp\bigl(f^{(r)}\bigr) \;\subseteq\; \widehat{\supp}\bigl(f^{(r)}\bigr)
\quad \text{for every sparse factor } f.
\]
This is the only structural property required for soundness:
cells outside $\widehat{\supp}$ are guaranteed to evaluate to zero, so
skipping them cannot remove any nonzero contribution to the round
message.

\begin{algorithm}[t]
\small
\caption{Sparsity-aware histogram-sumcheck round (one factor; multi-factor case intersects masks).}
\label{alg:sparse-sumcheck}
\begin{algorithmic}[1]
\State \textbf{Input.} Round index $r$; partially folded factor $f^{(r-1)}$; tracked mask $\widehat{\supp}(f^{(r-1)}) \subseteq \Bool^{n - r + 1}$; product factors $\{g^{(r-1)}_t\}_{t}$ that multiply $f$ in the round message; degree bound $D$.
\State \textbf{Output.} Univariate round-$r$ message $h_r \in \EFld[X]$ of degree $\leq D$.
\Statex
\State $h_r \gets 0$
\For{each $z \in \widehat{\supp}(f^{(r-1)})$ with $z_1 = 0$}
  \State \Comment{The first coordinate is the round-$r$ variable.}
  \State Let $z' = (z_2, \ldots, z_{n-r+1})$; for each $b \in \{0, 1\}$ and each evaluation point $X \in \{0, 1, \ldots, D\}$:
  \State \hskip1em $h_r(X) \mathrel{+}= f^{(r-1)}(X, z') \cdot \prod_t g^{(r-1)}_t(X, z')$ if $(X, z') \in \widehat{\supp}(f^{(r-1)})$ or $X \in \EFld \setminus \Bool$.
\EndFor
\State Send $h_r$. Receive challenge $\rho_r \in \EFld$ from the verifier.
\State $f^{(r)}(\cdot) \gets (1 - \rho_r) \cdot f^{(r-1)}(0, \cdot) + \rho_r \cdot f^{(r-1)}(1, \cdot)$
\State $\widehat{\supp}(f^{(r)}) \gets \{z' : (0, z') \in \widehat{\supp}(f^{(r-1)}) \text{ or } (1, z') \in \widehat{\supp}(f^{(r-1)})\}$
\State \Comment{Fold mask one variable ahead of the table.}
\If{$|\widehat{\supp}(f^{(r)})| \;>\; \theta \cdot 2^{n - r}$}
  \State Switch $f$ back to the dense path for remaining rounds. \Comment{Crossover threshold $\theta$ chosen so mask bookkeeping no longer pays for the zeros it skips.}
\EndIf
\State \Return $h_r$, updated $f^{(r)}$, $\widehat{\supp}(f^{(r)})$.
\end{algorithmic}
\end{algorithm}

\parh{Multi-factor masks.} When the round message contains a
product $\prod_t g_t \cdot f$ in which several factors are sparse,
the prover iterates over
$\bigcap_t \widehat{\supp}(g_t^{(r-1)}) \cap \widehat{\supp}(f^{(r-1)})$. The
intersection inherits the superset property factor-by-factor: any
$z$ outside the intersection lies outside at least one
$\widehat{\supp}$, hence outside the corresponding $\supp$, and the
product is zero at $z$.

\parh{Soundness.} The dense-path round-$r$ message is
\[
h_r^{\mathrm{dense}}(X)
  \;=\; \sum_{z \in \Bool^{n - r}} f^{(r-1)}(X, z) \prod_t g^{(r-1)}_t(X, z),
\]
and the sparse-path message restricts the outer sum to
$\widehat{\supp}(f^{(r-1)}) \cap \bigcap_t \widehat{\supp}(g^{(r-1)}_t)$. By the
support-superset invariant, every $z$ excluded from the
intersection has at least one factor evaluating to zero at
$(X, z)$ on $X \in \Bool$; for $X \notin \Bool$ the verifier
treats the round message as a degree-$D$ polynomial that is fully
determined by its values at $D + 1$ Boolean / small-integer
points, so the sparse and dense polynomials agree at every
evaluation point and therefore as polynomials. Hence
$h_r^{\mathrm{sparse}} \equiv h_r^{\mathrm{dense}}$ as univariate
polynomials over $\EFld$, and the sparse path emits the same
transcript as the dense path. The verifier's sumcheck soundness
is therefore inherited verbatim from the unmodified \zkls{}
protocol of \S\ref{sec:method:bucket}, with no additional error
term. \hfill$\square$

\parh{Crossover safety.} The threshold $\theta$ that triggers a
fall-back to the dense path is a prover-side bookkeeping decision:
the verifier observes only the round message $h_r$, which is
identical on both paths, so any choice of $\theta$ preserves
soundness. The threshold is tuned for prover-side performance
only.

\end{document}